\documentclass[a4paper,onecolumn,11pt]{quantumarticle}
\pdfoutput=1

\usepackage[font=small]{caption}
\usepackage{subcaption}

\usepackage{amsmath,amsfonts,amssymb,amsthm}
\usepackage[margin=1in]{geometry}
\usepackage[colorlinks]{hyperref}
\usepackage{mathrsfs}
\usepackage{mathtools}
\usepackage{thmtools}
\usepackage{thm-restate}
\mathtoolsset{centercolon}
\usepackage[numbers,comma,sort&compress]{natbib}
\usepackage[linesnumbered,ruled]{algorithm2e}
\usepackage{bm}
\usepackage{bookmark}
\usepackage{multirow}
\usepackage{xcolor}

\usepackage{enumerate}

\let\originalleft\left
\let\originalright\right
\renewcommand{\left}{\mathopen{}\mathclose\bgroup\originalleft}
\renewcommand{\right}{\aftergroup\egroup\originalright}

\renewcommand{\L}{\mathscr{L}}

\newcommand{\sgn}{\mathop{\mathrm{sgn}}}

\newcommand{\ket}[1]{|{#1}\rangle}
\newcommand{\bra}[1]{\langle{#1}|}

\DeclareMathOperator{\poly}{poly}

\newtheorem{theorem}{Theorem}
\newtheorem{lemma}{Lemma}
\newtheorem{corollary}{Corollary}
\newtheorem{problem}{Problem}
\newtheorem{proposition}{Proposition}
\newtheorem{remark}{Remark}

\numberwithin{equation}{section}

\renewcommand{\sec}[1]{\hyperref[sec:#1]{Section~\ref*{sec:#1}}}
\newcommand{\app}[1]{\hyperref[app:#1]{Appendix~\ref*{app:#1}}}
\newcommand{\thm}[1]{\hyperref[thm:#1]{Theorem~\ref*{thm:#1}}}
\newcommand{\lem}[1]{\hyperref[lem:#1]{Lemma~\ref*{lem:#1}}}
\newcommand{\cor}[1]{\hyperref[cor:#1]{Corollary~\ref*{cor:#1}}}
\newcommand{\prb}[1]{\hyperref[prb:#1]{Problem~\ref*{prb:#1}}}
\newcommand{\prop}[1]{\hyperref[prop:#1]{Proposition~\ref*{prop:#1}}}
\newcommand{\rem}[1]{\hyperref[rem:#1]{Remark~\ref*{rem:#1}}}
\newcommand{\tab}[1]{\hyperref[tab:#1]{Table~\ref*{tab:#1}}}
\newcommand{\fig}[1]{\hyperref[fig:#1]{Figure~\ref*{fig:#1}}}

\makeatletter
\let\@@magyar@captionfix\relax
\newcommand\notsotiny{\@setfontsize\notsotiny{7.5}{8.5}}
\makeatother

\begin{document}


\title{An Efficient Quantum Spectral-Galerkin Method for Linear Elliptic Equations}

\author[1]{Dongjie Wang}
\author[2]{Sheng Chen}
\affil[1]{School of Mathematical Sciences, Beijing Normal University, Beijing, 100875, China}
\affil[2]{Department of Mathematics, Faculty of Arts and Sciences, Beijing Normal University, Zhuhai, China}

\date{}
\maketitle

\begin{abstract}
This paper presents an efficient quantum spectral-Galerkin method for constant-coefficient second-order linear elliptic equations. The method employs a boundary-adapted Jacobi tensor-product basis that incorporates the boundary condition into the approximation space and produces Galerkin systems with polynomially bounded condition numbers and sparsity. For analytic solutions satisfying a suitable uniform regularity condition, spectral convergence allows the normalized Galerkin coefficient state to be prepared with complexity $\poly(d,\log(1/\epsilon))$, where $d$ is the spatial dimension and $\epsilon$ is the error tolerance, under standard sparse-access and state-preparation assumptions. At the state-preparation level, this improves the precision dependence over fixed-order finite-difference and finite-element discretizations. We further show how to estimate prescribed linear functionals of the solution, with an additional $O(1/\epsilon)$ factor arising from norm and inner-product estimation.
\end{abstract}

\medskip
\noindent\textbf{2020 Mathematics Subject Classification.}
Primary 65N35; Secondary 65N30, 68Q12, 81P68.

\noindent\textbf{Key words and phrases.}
Quantum spectral methods, spectral-Galerkin methods,
elliptic partial differential equations, quantum linear systems algorithms.
\section{Introduction}
\label{sec:introduction}
Partial differential equations (PDEs) provide fundamental mathematical models for diffusion, elasticity, electromagnetism, fluid mechanics, and quantum mechanics \cite{Evans10,Ciarlet88,Monk03,Temam01,Teschl09}. Despite their ubiquity, their numerical solution remains computationally challenging, especially in high dimensions. Two central difficulties are the curse of dimensionality \cite{Bellman57} and the cost of high-precision approximation. For a $d$-dimensional problem discretized with $N$ degrees of freedom per coordinate direction, the total number of unknowns scales as $N^d$, so even storing the full solution vector requires resources exponential in $d$. Moreover, for low-order discretization schemes, the approximation error typically decays only algebraically with $N$; hence achieving accuracy $\epsilon$ requires $N=\mathrm{poly}(1/\epsilon)$ before the cost of solving the linear system is considered.

Quantum computation offers a fundamentally different representation. Instead of storing all entries of a discretized solution vector, a quantum algorithm may prepare a normalized state whose amplitudes encode the solution coefficients. For a tensor-product discretization with $N$ one-dimensional degrees of freedom in each coordinate direction, the coefficient vector has length $N^d$, whereas the corresponding quantum state requires only $O(d\log N)$ qubits. This compact representation is one of the main motivations for quantum algorithms for PDEs. At the same time, it changes the computational objective: reading out many individual entries may be costly, so one typically extracts global quantities such as expectation values or linear functionals of the solution \cite{MP16}.

Quantum linear systems algorithms (QLSAs) provide a natural route from PDE discretizations to quantum algorithms \cite{HHL09,CKS17}. Many discretizations of linear PDEs reduce the computation of the expansion coefficients of an approximate solution to solving a finite-dimensional linear system $L\boldsymbol{c}=\boldsymbol{f}$. Given suitable access to $L$ and $\boldsymbol{f}$, a QLSA can prepare a quantum state proportional to the coefficient vector $\boldsymbol{c}$. The efficiency of this approach is governed by three basic quantities: the approximation error of the discretization, the condition number $\kappa(L)$, and the sparsity $s(L)$. The approximation error determines how large the discretization parameter $N$ must be to reach a prescribed tolerance $\epsilon$; the condition number controls the complexity and stability of the linear-system step; and the sparsity controls the cost of sparse matrix access. Therefore, a high-precision discretization suitable for QLSA-based algorithms should ideally combine spectral accuracy, $N=\poly(\log(1/\epsilon))$, with polynomially bounded condition number and sparsity, $\kappa(L)=\poly(N,d)$ and $s(L)=\poly(N,d)$.

\subsection{Related work}
\label{sec:related-work}

Several quantum frameworks for solving differential equations have been developed that are conceptually distinct from QLSA-based approaches. Schr\"odingerization reformulates general linear differential equations as Schr\"odinger-type dynamics through a warped phase transformation, thereby enabling their solution by Hamiltonian simulation \cite{JinLiuYu2024}. Another approach is the linear combination of Hamiltonian simulation (LCHS), which represents nonunitary dynamics as a linear combination of Hamiltonian simulations without relying on a reduction to a dilated linear-system problem \cite{AnLiuLin2023}. These approaches provide alternative routes to quantum differential equation solvers. In the following discussion of quantum PDE algorithms, we focus on QLSA-based methods, which are most directly related to the framework developed in this work.

\paragraph{QLSA-based quantum algorithms for PDEs.}
Low-order discretizations such as finite difference and finite element methods are well suited to sparse linear algebra, and quantum algorithms based on these discretizations have been studied in several works. For finite difference methods, quantum algorithms and explicit block encodings have been developed for various elliptic boundary value problems \cite{CLO21,Kharazi2025}. For finite element methods, Montanaro and Pallister \cite{MP16} analyzed quantum algorithms for estimating functionals of PDE solutions, building on classical finite element conditioning estimates such as those in \cite{BS89}. These approaches naturally produce sparse and structured linear systems. However, for fixed-order finite difference and finite element schemes, the approximation error decays only algebraically with the discretization parameter $N$. Consequently, achieving accuracy $\epsilon$ generally requires $N=\poly(1/\epsilon)$, so the resulting quantum complexity retains a polynomial dependence on $1/\epsilon$.

Spectral methods offer a natural way to overcome the algebraic accuracy barrier. For smooth or analytic solutions, they can converge exponentially fast, so high accuracy may be achieved with only $N=\poly(\log(1/\epsilon))$ degrees of freedom per coordinate direction \cite{STW11}. This makes spectral discretizations particularly attractive for high-precision quantum algorithms. However, extending quantum spectral methods to non-periodic boundary conditions introduces additional structural challenges. The Chebyshev pseudospectral method in \cite{CLO21} exploits the high-order convergence of Chebyshev expansions for Dirichlet problems, but the resulting differentiation matrices are generally non-Hermitian and non-normal. Consequently, their QLSA-relevant two-norm condition number is not determined by eigenvalue information alone, and analyzing it requires additional care. A recent alternative imposes Dirichlet boundary conditions through an antisymmetric extension followed by a Fourier-type periodic spectral method \cite{Febrianto2025}. This approach requires additional boundary compatibility conditions, including the vanishing of the relevant even-order derivatives at the boundary, which are not generally satisfied for a generic Dirichlet problem. Thus, obtaining a quantum spectral discretization that is simultaneously high-precision, well-conditioned, and sparse remains nontrivial.

\paragraph{Classical spectral-Galerkin methods.}

On the classical side, spectral-Galerkin methods based on Jacobi and Legendre polynomials are well established \cite{Shen94,STW11}. Shen \cite{Shen94} introduced efficient Legendre-Galerkin bases for second- and fourth-order problems using difference bases such as $L_{k+2}-L_k$. These constructions retain spectral accuracy and lead to efficient one-dimensional discretizations. In multidimensional tensor-product settings, however, the presence of a nontrivial mass matrix can substantially increase the sparsity of the resulting system matrix and may also deteriorate its conditioning as the dimension grows. These features are less favorable for QLSA-based algorithms, whose complexity depends directly on the sparsity and condition number of the linear system.

\paragraph{Summary of gaps.}

Low-order methods provide sparse linear systems but only algebraic accuracy, while existing spectral approaches face additional challenges in conditioning, boundary treatment, or high-dimensional tensor-product structure. This motivates the search for a spectral discretization that combines spectral accuracy with polynomially bounded condition number and sparsity.

\subsection{Our contributions}
\label{sec:contributions}

In this paper, we propose a quantum spectral-Galerkin method for $d$-dimensional constant-coefficient second-order linear elliptic boundary value problems with homogeneous Dirichlet boundary conditions. Rather than building the quantum algorithm on a collocation or pseudospectral discretization, we work with the weak formulation and construct a Dirichlet-compatible Jacobi-based spectral basis. The one-dimensional basis functions are $L^2$-orthonormal and vanish at the boundary, so the mass matrix is the identity, thereby avoiding the exponential-in-$d$ sparsity growth that can arise from nontrivial mass matrices in tensor-product spectral discretizations. The $d$-dimensional approximation space is obtained by taking tensor products of these basis functions, thereby incorporating the boundary condition directly into the approximation space.

This construction leads to symmetric positive-definite Galerkin systems whose condition numbers can be controlled through energy estimates. For analytic solutions satisfying a suitable uniform regularity condition, we prove spectral convergence, so that $N=\poly(\log(1/\epsilon))$ degrees of freedom per coordinate direction are sufficient. At the algebraic level, the Galerkin matrix satisfies $\kappa(L)=O(\kappa(A)N^4)$ and $s(L)=O(d^2N^2)$, where $A$ denotes the constant symmetric positive-definite coefficient matrix of the elliptic operator; when $A$ is diagonal, the sparsity improves to $O(dN)$. More generally, we observe that the underlying sparsity and conditioning mechanisms are not specific to homogeneous Dirichlet boundary conditions: they extend to other boundary conditions whenever a suitable $L^2$-orthonormal polynomial basis is available, the associated variational problem remains coercive, and the corresponding approximation space satisfies appropriate inverse estimates. We further introduce a rotated Jacobi-based basis that preserves the approximation space and condition number while reducing the sparsity to $O(dN+d^2)$ in the general case and to $1$ when $A$ is diagonal.

Combining these properties with a sparse-input QLSA yields a quantum algorithm for preparing the normalized Galerkin coefficient state with complexity $\poly(d,\log(1/\epsilon))$, under standard sparse-access and state-preparation assumptions. At the state-preparation level, this improves the precision dependence over fixed-order finite-difference and finite-element discretizations. We further show how to estimate prescribed linear functionals of the PDE solution from the prepared state, with an additional $O(1/\epsilon)$ factor arising from norm and inner-product estimation.

\tab{intro-comparison} summarizes the approximation accuracy, condition number, sparsity, and quantum state-preparation complexity of representative low-order discretizations and the spectral-Galerkin method proposed in this paper. \fig{condition-comparison} compares the two-norm condition-number growth of the Chebyshev pseudospectral and Jacobi-based spectral-Galerkin discretizations for $d=3$ and $d=4$, showing substantially milder growth for the latter over the tested range of $N$.

\begin{table}[t]
\begin{center}
\small
\renewcommand{\arraystretch}{1.25}
\resizebox{0.98\textwidth}{!}{
\begin{tabular}{|c|c|c|c|c|}
\hline
\textbf{Method}
&
\textbf{Error scaling}
&
\textbf{Condition number}
&
\textbf{Sparsity}
&
\textbf{Complexity}
\\
\hline
FDM \cite{CLO21,Kharazi2025}
&
$N^{-p}$
&
$O(N^2)$
&
$O(pd)$
&
$\poly(d,1/\epsilon)$
\\
\hline
FEM \cite{MP16,BS89}
&
$N^{-p}$
&
$O(N^2)$
&
$O(1)$
&
$\poly(d,1/\epsilon)$
\\
\hline
Spectral-Galerkin method [this paper]
&
$(C/N)^N$
&
$O(N^4)$
&
$O(dN)$
&
$\poly(d,\log(1/\epsilon))$
\\
\hline
\end{tabular}
}
\end{center}
\caption{
Comparison of discretization properties relevant to QLSA-based algorithms for the $d$-dimensional Poisson equation with homogeneous Dirichlet boundary conditions. Here $N$ denotes the number of degrees of freedom in each coordinate direction, $p$ denotes the order of a fixed-order finite difference or finite element scheme, and the last column gives the complexity of the state-preparation part of the quantum algorithm. For FEM, the sparsity $O(1)$ reported in \cite{MP16} is understood under the standard convention that the spatial dimension $d$ is fixed; if the dependence on $d$ is tracked explicitly, the sparsity generally depends on $d$.
}
\label{tab:intro-comparison}
\end{table}

\begin{figure}[t]
\centering
\includegraphics[width=\textwidth]{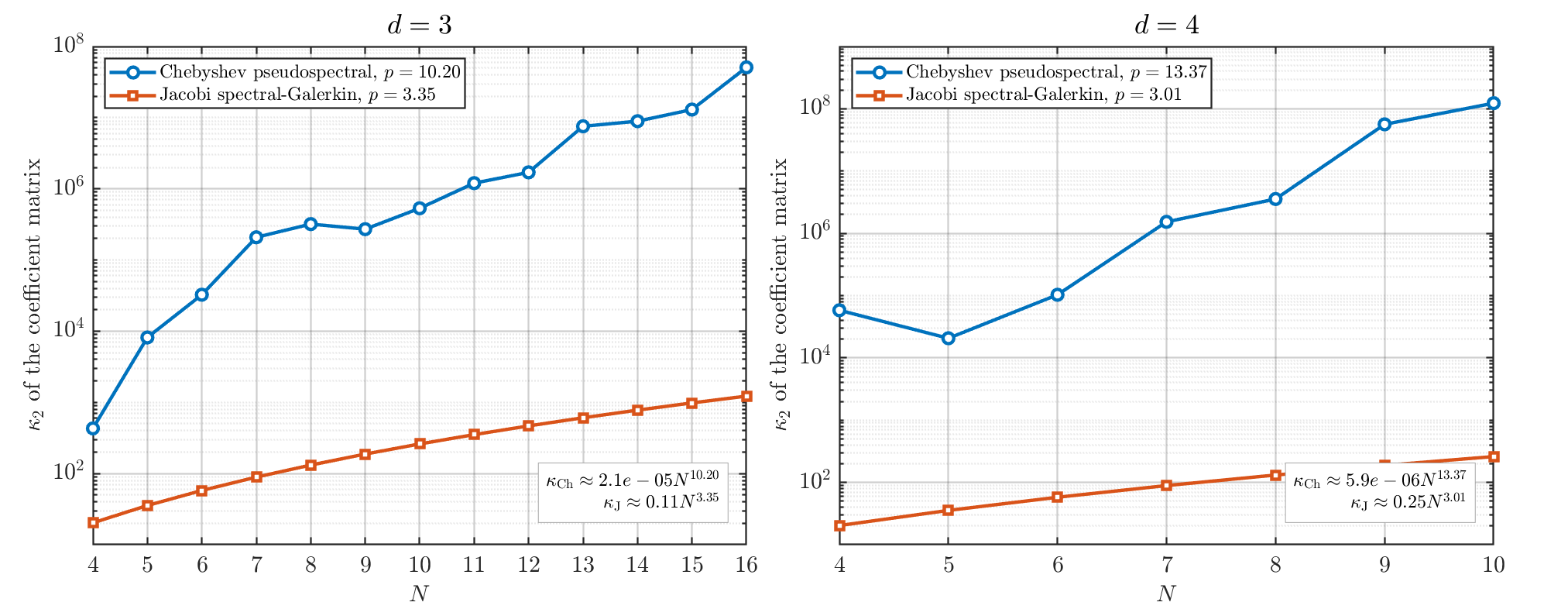}
\caption{
Numerical comparison of the QLSA-relevant two-norm condition numbers of the Chebyshev pseudospectral coefficient matrix \cite{CLO21} and our Jacobi-based spectral-Galerkin coefficient matrix for the $d$-dimensional Poisson equation, with $d=3$ and $d=4$. Here $N$ denotes the number of one-dimensional degrees of freedom. The fitted exponent $p$ corresponds to the power-law model $\kappa_2\approx C N^p$ over the indicated range of $N$. These results are intended as an illustrative finite-dimensional comparison rather than an asymptotic estimate.
}
\label{fig:condition-comparison}
\end{figure}

The remainder of this paper is organized as follows. \sec{model problem and algorithmic objectives} states the model problem and the algorithmic objectives. \sec{galerkin method and challenges} derives the Galerkin discretization and discusses the challenges associated with conventional basis functions. \sec{jacobi-spectral-Galerkin method} presents our new Jacobi-based spectral-Galerkin method and analyzes its properties and classical complexity. \sec{quantum spectral-Galerkin method} develops the corresponding quantum algorithm. Finally, \sec{conclusion} concludes with a discussion of the limitations and open problems.

\section{Model problem and algorithmic objectives}
\label{sec:model problem and algorithmic objectives}

This section introduces the model elliptic problem and its weak formulation, and then specifies the sparse-input QLSA framework and the algorithmic objectives that guide our discretization.

\subsection{Model problem and weak formulation}
\label{sec:model problem}

In this paper, we consider the following model problem.

\begin{problem}\label{prb:pde}
Given a $d$-dimensional second-order linear elliptic boundary value problem with homogeneous Dirichlet boundary conditions,
\begin{equation}
\L u
=
-\sum_{p,q=1}^{d}
A_{pq}\,\partial_{x_p}\partial_{x_q}u
=
f
\quad \text{in } \Omega:=(-1,1)^d,
\qquad
u|_{\partial\Omega}=0,
\label{eq:model-problem}
\end{equation}
where $f\in C^\infty(\overline{\Omega})$, and
$A=(A_{pq})_{p,q=1}^d\in\mathbb{R}^{d\times d}$ is a constant symmetric positive-definite matrix. Our goal is to design a quantum algorithm to estimate a linear functional of $u$: for a given function $\rho\in L^2(\Omega)$, estimate
\begin{equation}
\langle \rho,u\rangle
:=
\int_\Omega \rho(\boldsymbol{x})u(\boldsymbol{x})\,\mathrm{d}\boldsymbol{x}.
\end{equation}
\end{problem}

The weak formulation of \eqref{eq:model-problem} seeks $u\in H_0^1(\Omega)$ such that
\begin{equation}
a(u,v)=F(v),
\qquad
\forall\,v\in H_0^1(\Omega),
\label{eq:weak-form}
\end{equation}
where the bilinear form $a(\cdot,\cdot)$ is defined by
\begin{equation}
a(u,v)
:=
\sum_{p,q=1}^{d}
\int_\Omega
A_{pq}\,
\partial_{x_p}u\,
\partial_{x_q}v
\,\mathrm{d}\boldsymbol{x}
=
\int_\Omega
(\nabla u)^{\mathrm T}
(A\nabla v)
\,\mathrm{d}\boldsymbol{x},
\label{eq:bilinear-form}
\end{equation}
and the linear functional $F(\cdot)$ is defined by
\begin{equation}
F(v)
:=
\langle f,v\rangle
=
\int_\Omega
f(\boldsymbol{x})v(\boldsymbol{x})
\,\mathrm{d}\boldsymbol{x}.
\label{eq:rhs-functional}
\end{equation}
Since $A$ is symmetric positive definite, the bilinear form $a(\cdot,\cdot)$ is continuous and coercive on $H_0^1(\Omega)$. By the Lax--Milgram theorem \cite{Evans10}, \eqref{eq:model-problem} admits a unique weak solution $u\in H_0^1(\Omega)$.

\begin{remark}
Following \cite{MP16}, we take the estimation of a linear functional of the solution to be the target computational task. A quantum algorithm for a PDE typically provides access to the solution in an encoded form rather than to all pointwise values of $u$. Therefore, recovering the full solution is generally not the natural computational objective, whereas linear functionals represent natural global quantities of interest.

In addition, this choice is consistent with the standard $L^2$ error criterion. Indeed, if $\tilde{u}$ is an approximation to $u$, then the Cauchy--Schwarz inequality gives
\begin{equation}
\bigl|\langle \rho,\tilde{u}\rangle-\langle \rho,u\rangle\bigr|
=
\bigl|\langle \rho,\tilde{u}-u\rangle\bigr|
\le
\|\rho\|_{L^2(\Omega)}
\|\tilde{u}-u\|_{L^2(\Omega)}.
\end{equation}
\end{remark}

In contrast to \cite{CLO21}, we do not require the above PDE to satisfy the \emph{global strict diagonal dominance condition}
\begin{equation}
C
:=
1-
\sum_{p=1}^{d}
\frac{1}{|A_{pp}|}
\sum_{q\in[d]\setminus\{p\}}
|A_{pq}|
>0,
\label{eq:gsdd}
\end{equation}
where $[d]:=\{1,2,\dots,d\}$. This assumption is not needed in our subsequent condition number analysis. For simplicity, we focus on the homogeneous Dirichlet case. The nonhomogeneous Dirichlet problem with boundary data
\begin{equation}
u|_{\partial\Omega}=\gamma,
\qquad
\gamma\in C^\infty(\partial\Omega),
\end{equation}
can be reduced to the homogeneous case by a standard change of variables, with the boundary contribution absorbed into the right-hand side.

\subsection{Sparse-input QLSA and algorithmic objectives}
\label{sec:qlsa and objectives}

A numerical discretization of \eqref{eq:model-problem} leads to a finite-dimensional linear system. Since the complexity of a quantum linear systems algorithm depends on how the coefficient matrix and the right-hand side are accessed, we adopt the sparse-input model rather than assuming that a block encoding of the coefficient matrix is given directly. We recall the corresponding QLSA result below.

\begin{lemma}[Sparse-input QLSA, adapted from \cite{CKS17}]\label{lem:qlsa}
Let $H \in \mathbb{C}^{\mathcal{N}\times \mathcal{N}}$ be a Hermitian, invertible matrix with condition number $\kappa=\|H\|\,\|H^{-1}\|$, and suppose that $H$ is $s$-sparse, namely, each row and each column of $H$ contains at most $s$ nonzero entries. Let $b\in\mathbb{C}^\mathcal{N}$ be a nonzero vector, and define $x=H^{-1}b$. Let
\begin{equation}
    |b\rangle=\frac{b}{\|b\|},\qquad |x\rangle=\frac{x}{\|x\|}.
\end{equation}
Assume that we are given an oracle $P_H$ which, on input $(r,\ell)$, returns the column index together with the value of the $\ell$th nonzero entry in row $r$, and an oracle $P_b$ that prepares the state $|b\rangle$. Then, for any $\epsilon>0$, there exists a quantum algorithm that prepares a state $\epsilon$-close to $|x\rangle$ in $\ell_2$ norm, with success probability $\Omega(1)$, using
\begin{equation}
    O\left(s\kappa\,\poly(\log(s\kappa/\epsilon))\right)
\end{equation}
queries to $P_H$ and $P_b$. Moreover, the gate complexity exceeds the query complexity by at most a factor of $\poly(\log \mathcal{N})$.
\end{lemma}

Guided by \lem{qlsa}, our complexity analysis focuses on three quantities: the dimension of the linear system, which is determined by the approximation accuracy of the discretization; the condition number of the coefficient matrix; and its sparsity. Accordingly, we seek a discretization with \textbf{rapid convergence}, a \textbf{polynomially bounded condition number}, and \textbf{polynomial sparsity}.
\section{Galerkin method and challenges}
\label{sec:galerkin method and challenges}

In this section, we first formulate a general tensor-product Galerkin discretization of \eqref{eq:model-problem} and express the resulting coefficient matrix in terms of the one-dimensional mass, stiffness, and mixed derivative matrices. We then examine two illustrative choices of basis functions. These examples reveal the trade-offs among approximation accuracy, condition number, and sparsity, and motivate the search for a basis that simultaneously achieves rapid convergence, a polynomially bounded condition number, and polynomial sparsity.

\subsection{Tensor-product Galerkin discretization}
\label{sec:galerkin method}

We discretize the weak formulation of \eqref{eq:model-problem} by restricting it to a finite-dimensional subspace of $H_0^1(\Omega)$. More precisely, for a finite-dimensional subspace $V_N\subset H_0^1(\Omega)$, the Galerkin approximation $u_N\in V_N$ is defined by
\begin{equation}
a(u_N,v_N)=F(v_N), \qquad \forall v_N\in V_N.
\label{eq:galerkin-discrete}
\end{equation}

We choose $V_N$ to be a tensor-product subspace of $H_0^1(\Omega)$. Let $[N]_0:=\{0,1,\dots,N-1\}$, and let $\{\phi_k\}_{k\in [N]_0}\subset H_0^1(-1,1)$ be a family of one-dimensional basis functions. For each multi-index $\boldsymbol{k}=(k_1,\dots,k_d)\in [N]_0^d$, define
\begin{equation}
\Phi_{\boldsymbol{k}}(\boldsymbol{x})=\prod_{j=1}^d \phi_{k_j}(x_j),
\quad \boldsymbol{x}=(x_1,\dots,x_d)\in \Omega.
\label{eq:tensor-basis}
\end{equation}
We then set
\begin{equation}
V_N:=\mathrm{span}\bigl\{\Phi_{\boldsymbol{k}}\mid\boldsymbol{k}\in [N]_0^d\bigr\}\subset H_0^1(\Omega).
\label{eq:VN-space}
\end{equation}
To facilitate the quantum implementation, we assume that $N=2^n$, so that $\dim(V_N)=N^d=2^{nd}$. With this choice of basis, the Galerkin solution admits the expansion
\begin{equation}
u_N(\boldsymbol{x})=\sum_{\boldsymbol{k}\in [N]_0^d} c_{\boldsymbol{k}}\,\Phi_{\boldsymbol{k}}(\boldsymbol{x}),
\label{eq:uN-expansion}
\end{equation}
where the coefficients $\{c_{\boldsymbol{k}}\}$ are determined by \eqref{eq:galerkin-discrete}. 

To derive the algebraic system for the coefficients
$\{c_{\boldsymbol{k}}\}_{\boldsymbol{k}\in [N]_0^d}$, we substitute the expansion
\eqref{eq:uN-expansion} into \eqref{eq:galerkin-discrete} and choose the test
function $v_N=\Phi_{\boldsymbol{l}}$ for each $\boldsymbol{l}\in [N]_0^d$.
Then \eqref{eq:galerkin-discrete} is equivalent to
\begin{equation}
\sum_{p,q=1}^d A_{pq}\sum_{\boldsymbol{k}\in[N]_0^d}c_{\boldsymbol{k}}
\bigl\langle \partial_{x_p}\Phi_{\boldsymbol{k}},\partial_{x_q}\Phi_{\boldsymbol{l}}\bigr\rangle
=
\langle f,\Phi_{\boldsymbol{l}}\rangle,
\qquad \forall\,\boldsymbol{l}\in[N]_0^d.
\label{eq:discrete-entrywise}
\end{equation}

To write this system in a compact form, we introduce the coefficient vector and the load vector by
\begin{equation}
\boldsymbol{c}:=\sum_{\boldsymbol{k}\in[N]_0^d}c_{\boldsymbol{k}}\ket{\boldsymbol{k}},
\qquad
\boldsymbol{f}:=\sum_{\boldsymbol{l}\in[N]_0^d}\langle f,\Phi_{\boldsymbol{l}}\rangle\ket{\boldsymbol{l}},
\label{eq:c-f-def}
\end{equation}
and define the one-dimensional matrices
\begin{equation}
M:=\sum_{j,k\in[N]_0}\langle\phi_k,\phi_j\rangle\ket{j}\bra{k},
\quad
S:=\sum_{j,k\in[N]_0}\langle\phi_k',\phi_j'\rangle\ket{j}\bra{k},
\quad
T:=\sum_{j,k\in[N]_0}\langle\phi_k',\phi_j\rangle\ket{j}\bra{k}.
\label{eq:MST-def}
\end{equation}
Here $M$, $S$, and $T$ are
respectively the mass matrix, the stiffness matrix, and the mixed derivative
matrix associated with the one-dimensional basis
$\{\phi_k\}_{k\in[N]_0}$.

By the tensor-product structure of the basis functions, \eqref{eq:discrete-entrywise} is equivalent to the linear system
\begin{equation}
L\boldsymbol{c}=\boldsymbol{f},
\label{eq:linear-system}
\end{equation}
where $L=L_1+L_2$ with
\begin{align}
L_1
&=
\sum_{p=1}^d
A_{pp}
\bigl(
M^{\otimes(p-1)}
\otimes S
\otimes M^{\otimes(d-p)}
\bigr),
\label{eq:L1-final}
\\
L_2
&=
-2
\sum_{1\le p<q\le d}
A_{pq}
\bigl(
M^{\otimes(p-1)}
\otimes T
\otimes M^{\otimes(q-p-1)}
\otimes T
\otimes M^{\otimes(d-q)}
\bigr).
\label{eq:L2-final}
\end{align}
The detailed derivation of \eqref{eq:L1-final}--\eqref{eq:L2-final} is given in \app{derivation}. Solving \eqref{eq:linear-system} yields the coefficient vector $\boldsymbol{c}$ and hence the Galerkin approximation $u_N$.

\subsection{Two illustrative choices of basis functions}\label{sec:two choices}
In view of \lem{qlsa}, two properties of the coefficient matrix $L$ are particularly important for the QLSA complexity: its condition number $\kappa(L)$ and its sparsity $s(L)$, where $s(L)$ denotes the maximum number of nonzero entries in any row or column of $L$. Since the QLSA complexity depends essentially linearly on these quantities, it is important to understand how they are affected by the choice of basis. To illustrate this dependence, we consider two classical examples. Throughout this subsection, we restrict ourselves to the case $A=I_d$, for which the off-diagonal part vanishes and the system matrix reduces to
\begin{equation}
L=L_1=\sum_{p=1}^d M^{\otimes(p-1)}\otimes S\otimes M^{\otimes(d-p)}.
\label{eq:L-example}
\end{equation}

\paragraph{Example 1: piecewise polynomial basis \cite{BS89}.}
Suppose that the one-dimensional basis functions $\{\phi_k\}_{k\in[N]_0}$ are chosen to be continuous piecewise polynomials on a partition of $(-1,1)$. Then the resulting tensor-product Galerkin discretization is precisely a finite element method. Since low-order finite element bases only provide algebraic convergence, in order to achieve an approximation error of order $\epsilon$, one generally needs
\begin{equation}
N=\poly(1/\epsilon).
\label{eq:fem-N-eps}
\end{equation}
For such finite element bases, the condition number of the coefficient matrix satisfies
\begin{equation}
\kappa(L)=O\bigl((\dim V_N)^{2/d}\bigr)=O(N^2).
\label{eq:fem-kappa}
\end{equation}

In particular, if $\phi_k$ are piecewise linear functions, then the one-dimensional mass matrix $M$ and stiffness matrix $S$ are both tridiagonal. Consequently, each tensor-product term in \eqref{eq:L-example} has at most $3^d$ nonzero entries in each row and column, and hence
\begin{equation}
s(L)=O(d\,3^d).
\label{eq:fem-sparsity}
\end{equation}

Therefore, by \eqref{eq:fem-kappa}, \eqref{eq:fem-sparsity}, and \lem{qlsa}, the query complexity of solving \eqref{eq:linear-system} by QLSA is
\begin{equation}
O\left(d\,3^dN^2\,\poly\!\left(\log(d\,3^dN^2/\delta)\right)\right),
\label{eq:fem-qlsa-complexity}
\end{equation}
where $\delta$ is the target precision of the quantum linear solver. Thus, in this finite element setting, the QLSA complexity remains polynomial in $1/\epsilon$ and, under the present sparse-matrix model, grows exponentially with the dimension $d$ because of the factor $3^d$.

\paragraph{Example 2: Legendre difference basis \cite{Shen94}.}
Consider the one-dimensional basis
\begin{equation}
    \phi_k(x)=c_k\bigl(L_{k+2}(x)-L_k(x)\bigr), \qquad k\in[N]_0,
\end{equation}
where $L_k$ is the Legendre polynomial of degree $k$, and the normalization constants $c_k$ are chosen such that $S=I$. This is the classical Legendre--Galerkin basis introduced in \cite{Shen94}. Moreover, for smooth solutions, the associated approximation enjoys spectral accuracy, so that to achieve an approximation error of order $\epsilon$, one may take
\begin{equation}
N=\poly(\log(1/\epsilon)).
\label{eq:legendre-N-eps}
\end{equation}
In this case, the coefficient matrix is
\begin{equation}
L=\sum_{p=1}^d M^{\otimes(p-1)}\otimes I\otimes M^{\otimes(d-p)}. 
\label{eq:L-legendre}
\end{equation}
Since \(M\) is symmetric positive definite, 
\begin{equation}
\kappa(L)=\frac{\lambda_{\text{max}}(L)}{\lambda_{\text{min}}(L)}=\frac{d\lambda_{\text{max}}(M)^{d-1}}{d\lambda_{\text{min}}(M)^{d-1}}=\kappa(M)^{d-1}.
\label{eq:legendre-kappa}
\end{equation}

Moreover, the mass matrix $M$ is symmetric positive definite and satisfies $M_{ij}\neq 0$ only when $i=j$ or $i=j\pm2$. Consequently, each tensor-product term in \eqref{eq:L-legendre} has at most $3^{d-1}$ nonzero entries in each row and column, and hence
\begin{equation}
s(L)=O(d\,3^{d-1})=O(d\,3^d).
\label{eq:legendre-sparsity}
\end{equation}

Therefore, by \eqref{eq:legendre-kappa}, \eqref{eq:legendre-sparsity}, and \lem{qlsa}, the query complexity of solving \eqref{eq:linear-system} by QLSA is
\begin{equation}
O\left(d\,3^d\,\kappa(M)^{d-1}\,\poly\!\left(\log(d\,3^d\,\kappa(M)^{d-1}/\delta)\right)\right),
\label{eq:legendre-qlsa-complexity}
\end{equation}
where $\delta$ is the target precision of the quantum linear solver. Thus, although this Legendre basis can achieve spectral accuracy and hence requires only $N=\poly(\log(1/\epsilon))$, under the present sparse-matrix model the QLSA complexity still grows exponentially with the dimension $d$, through both the sparsity factor $3^d$ and the condition-number factor $\kappa(M)^{d-1}$.

The above two examples show that the choice of basis is crucial. Local low-order bases may lead to better-conditioned systems, but their approximation power is limited, so achieving high accuracy requires a large value of $N$. On the other hand, spectrally accurate bases can significantly reduce the required value of $N$, but the resulting matrices may become poorly conditioned as the dimension grows. Moreover, in both examples, the sparsity of the resulting coefficient matrix grows exponentially with the dimension $d$. 

This motivates the following question: can one construct a family of basis functions such that the resulting Galerkin method simultaneously has the following three properties?
\begin{enumerate}
    \item The approximation error decays rapidly with $N$, ideally so that achieving accuracy $\epsilon$ only requires $N=\poly(\log(1/\epsilon))$, rather than $N=\poly(1/\epsilon)$.
    \item The condition number of the coefficient matrix satisfies $\kappa(L)=\poly(N,d)$.
    \item The sparsity of the coefficient matrix satisfies $s(L)=\poly(N,d)$.
\end{enumerate}

\section{An \texorpdfstring{$L^2$}{L2}-orthonormal Jacobi-based spectral-Galerkin method}
\label{sec:jacobi-spectral-Galerkin method}

Motivated by the difficulties identified in \sec{two choices}, we now introduce a Jacobi-based spectral-Galerkin method for \eqref{eq:model-problem}. The proposed basis incorporates the homogeneous Dirichlet boundary condition directly while preserving orthonormality in the standard $L^2$ inner product. We show that the resulting discretization has a condition number and sparsity polynomial in $N$ and $d$, while spectral accuracy is obtained for analytic solutions satisfying a suitable uniform regularity condition.

We first present the Jacobi-based basis construction in \sec{basis} and establish its main approximation, condition number, and sparsity properties in \lem{jacobi-basis-properties}. We then introduce a rotated Jacobi-based basis in \sec{rotated-basis}, which preserves the approximation space and condition number while further reducing the sparsity of the coefficient matrix. Finally, in \sec{classical complexity}, we analyze the classical complexity of solving the resulting linear systems by the conjugate gradient method. 

Basic properties of Jacobi polynomials used in the analysis are collected in \app{Jacobi}, while a concrete example and numerical results are presented in \app{example}.

\subsection{Jacobi-based basis and its properties}\label{sec:basis}

Motivated by the three properties stated above, we now explain how our basis is constructed. To achieve property~1, namely, spectral accuracy, it is natural to look for basis functions related to orthogonal polynomials on $(-1,1)$. A natural candidate is provided by Jacobi polynomials, which include Legendre and Chebyshev polynomials as special cases. For parameters $\alpha,\beta>-1$, the Jacobi polynomials $\{P_k^{(\alpha,\beta)}\}_{k\ge0}$ are orthogonal with respect to the weight $\omega^{(\alpha,\beta)}(x)=(1-x)^\alpha(1+x)^\beta$, and satisfy
\begin{equation}
\label{eq:jacobi-ortho-subsection}
\int_{-1}^1
P_k^{(\alpha,\beta)}(x)\,P_l^{(\alpha,\beta)}(x)\,
\omega^{(\alpha,\beta)}(x)\,dx
=
\gamma_k^{(\alpha,\beta)}\delta_{kl},
\end{equation}
where
\begin{equation}
\gamma_k^{(\alpha,\beta)}
=
\frac{2^{\alpha+\beta+1}\Gamma(k+\alpha+1)\Gamma(k+\beta+1)}
{(2k+\alpha+\beta+1)k!\,\Gamma(k+\alpha+\beta+1)}.
\end{equation}

The main difficulty is to impose the homogeneous Dirichlet boundary condition while keeping the resulting coefficient matrix well structured. In Example~2, this was done by taking the Legendre difference basis $L_{k+2}-L_k$. Although that choice preserves spectral accuracy, it leads to a nontrivial mass matrix $M$, and the sparsity of the resulting tensor-product system is then dominated by $M$, which in turn causes exponential growth with the dimension $d$.

This observation suggests a different strategy. Instead of enforcing the boundary condition by taking differences of orthogonal polynomials, we incorporate the boundary factor directly into the basis functions and define
\begin{equation}
\phi_k(x):=\frac{1}{\sqrt{\gamma_k^{(2,2)}}}(1-x^2)P_k^{(2,2)}(x),
\quad k\in[N]_0.
\label{eq:new-basis-1d}
\end{equation}
Then each $\phi_k$ vanishes at $x=\pm1$, so the homogeneous Dirichlet boundary condition is satisfied automatically. Moreover, by \eqref{eq:jacobi-ortho-subsection}, the family $\{\phi_k\}_{k\in[N]_0}$ is orthonormal in $L^2(-1,1)$, and therefore
\begin{equation}
M=I.
\label{eq:new-basis-MI}
\end{equation}
Consequently, the matrices $L_1$ and $L_2$ take the form
\begin{align}
L_1&=\sum_{p=1}^d A_{pp}\bigl(I^{\otimes(p-1)}\otimes S\otimes I^{\otimes(d-p)}\bigr),
\label{eq:new-basis-L1}\\
L_2&=-2\sum_{1\le p<q\le d}A_{pq}
\bigl(I^{\otimes(p-1)}\otimes T\otimes I^{\otimes(q-p-1)}\otimes T\otimes I^{\otimes(d-q)}\bigr).
\label{eq:new-basis-L2}
\end{align}
Hence the sparsity of the full coefficient matrix $L=L_1+L_2$ is determined entirely by the one-dimensional matrices $S$ and $T$, rather than by tensor powers of a nontrivial mass matrix $M$. This removes the main source of exponential sparsity growth in the previous examples and is therefore the key step toward achieving property~3.

Finally, we emphasize that our Galerkin formulation is based on the standard unweighted inner product, that is, the weight is $1$. One may also consider weighted Galerkin formulations associated with other Jacobi weights. In such settings, it is still possible to construct bases with $M=I$. However, the resulting one-dimensional matrices $S$ and $T$ are in general no longer symmetric and skew-symmetric, respectively, and consequently the full coefficient matrix $L=L_1+L_2$ may lose its favorable structure and become poorly conditioned in high dimensions. By contrast, in the present unweighted setting, $S$ is symmetric and $T$ is skew-symmetric. Since $A$ is symmetric, it follows that both $L_1$ and $L_2$ are symmetric, and hence so is $L$. This structural property will be crucial for proving the favorable condition number bound later.

To establish property~2, we first estimate the condition number of the one-dimensional stiffness matrix $S$ associated with the basis \eqref{eq:new-basis-1d}.

\begin{lemma}\label{lem:S-cond-N4}
Let $\{\phi_k\}_{k\in[N]_0}$ be the one-dimensional basis defined in \eqref{eq:new-basis-1d}, and let $S$ be the associated stiffness matrix defined in \eqref{eq:MST-def}. Then
\begin{equation}
\lambda_{\min}(S)\ge \frac{\pi^2}{4},
\qquad
\lambda_{\max}(S)\le \frac{1}{4}(N+2)^2(N+3)^2.
\end{equation}
Consequently, $\kappa(S)=O(N^4)$.
\end{lemma}

\begin{proof}
Let $X_N:=\mathrm{span}\{\phi_0,\phi_1,\dots,\phi_{N-1}\}$.
Since $\{\phi_k\}_{k\in[N]_0}$ is an orthonormal basis of $X_N$ in $L^2(-1,1)$, for any
\begin{equation}
v_N=\sum_{k=0}^{N-1} c_k\phi_k\in X_N
\end{equation}
with coefficient vector $\boldsymbol{c}=(c_0,\dots,c_{N-1})^\mathrm{T}$, we have
\begin{equation}
\|v_N\|_{L^2(-1,1)}^2=\boldsymbol{c}^\mathrm{T}\boldsymbol{c},
\qquad
\|v_N'\|_{L^2(-1,1)}^2=\boldsymbol{c}^\mathrm{T} S\boldsymbol{c}.
\end{equation}
Therefore,
\begin{equation}
\lambda_{\min}(S)=\min_{0\neq v_N\in X_N}\frac{\|v_N'\|_{L^2(-1,1)}^2}{\|v_N\|_{L^2(-1,1)}^2},
\qquad
\lambda_{\max}(S)=\max_{0\neq v_N\in X_N}\frac{\|v_N'\|_{L^2(-1,1)}^2}{\|v_N\|_{L^2(-1,1)}^2}.
\end{equation}

Since $v_N\in H_0^1(-1,1)$, the Poincar\'e inequality \cite{PW60} yields
\begin{equation}
\|v_N\|_{L^2(-1,1)}\le \frac{2}{\pi}\|v_N'\|_{L^2(-1,1)}.
\end{equation}
Therefore,
\begin{equation}
\lambda_{\min}(S)\ge \frac{\pi^2}{4}.
\end{equation}

On the other hand, each $v_N$ is a polynomial of degree at most $N+1$. Thus, by the inverse inequality for algebraic polynomials on $[-1,1]$ \cite{STW11},
\begin{equation}
\|v_N'\|_{L^2(-1,1)}\le \frac{1}{2}(N+2)(N+3) \|v_N\|_{L^2(-1,1)}.
\end{equation}
Therefore,
\begin{equation}
\lambda_{\max}(S)\le \frac{1}{4}(N+2)^2(N+3)^2.
\end{equation}

Combining the above bounds, we obtain
\begin{equation}
\kappa(S)=\frac{\lambda_{\max}(S)}{\lambda_{\min}(S)}
\le \frac{1}{\pi^2}(N+2)^2(N+3)^2=O(N^4),
\end{equation}
which proves the claim.
\end{proof}

The next lemma shows that, for the Kronecker-sum structure arising here, the condition number of the multidimensional matrix is completely determined by that of the one-dimensional stiffness matrix.

\begin{lemma}\label{lem:S}
Let $S\in\mathbb{R}^{N\times N}$ be symmetric positive definite, and define
\begin{equation}
L_0:=\sum_{p=1}^d I^{\otimes(p-1)}\otimes S\otimes I^{\otimes(d-p)}.
\label{eq:L-kron-sum}
\end{equation}
Then
\begin{equation}
\lambda_{\max}(L_0)=d\,\lambda_{\max}(S),
\qquad
\lambda_{\min}(L_0)=d\,\lambda_{\min}(S),
\label{eq:L-extreme-eigs}
\end{equation}
and hence
\begin{equation}
\kappa(L_0)=\kappa(S).
\label{eq:L-cond-S}
\end{equation}
\end{lemma}

\begin{proof}
Let $\{v_i\}_{i\in[N]_0}$ be an orthonormal eigenbasis of $S$, and let
$\{\lambda_i\}_{i\in[N]_0}$ be the corresponding positive eigenvalues, so that $Sv_i=\lambda_i v_i,\,i\in[N]_0$.
Then for any $i_1,\dots,i_d\in[N]_0$, the tensor-product vector
$v_{i_1}\otimes\cdots\otimes v_{i_d}$ is an eigenvector of $L_0$ with eigenvalue
\begin{equation}
\lambda_{i_1}+\cdots+\lambda_{i_d}.
\end{equation}
Since each $\lambda_i$ is positive, it follows that
\begin{equation}
\lambda_{\max}(L_0)=d\,\lambda_{\max}(S),
\qquad
\lambda_{\min}(L_0)=d\,\lambda_{\min}(S).
\end{equation}
Therefore,
\begin{equation}
\kappa(L_0)=\frac{\lambda_{\max}(L_0)}{\lambda_{\min}(L_0)}
=\frac{d\,\lambda_{\max}(S)}{d\,\lambda_{\min}(S)}
=\kappa(S).
\end{equation}
This completes the proof.
\end{proof}

Combining the two lemmas above, we obtain the following lemma, which shows that the proposed basis indeed satisfies the three desired properties.

\begin{lemma}\label{lem:jacobi-basis-properties}
Consider an instance of \prb{pde}, namely, the $d$-dimensional second-order linear elliptic boundary value problem with homogeneous Dirichlet boundary conditions \eqref{eq:model-problem}. Let the approximation space be generated by the tensor-product basis associated with the one-dimensional basis \eqref{eq:new-basis-1d}. Then this basis satisfies the three properties stated above in the following sense:
\begin{enumerate}
    \item Suppose that the exact solution is analytic and satisfies the uniform regularity condition
    \begin{equation}
        C_d:=\sup_{N\in\mathbb{N}} C_N<\infty,
        \label{eq:Cd}
    \end{equation}
    where
    \begin{equation}
        C_N
        :=
        \left(
        \sum_{i=1}^d
        \left(
        \sum_{j=1}^d
        \|\partial_{x_j}^{\,N+1}\partial_{x_i}u\|_{W_{i,j}^{(N+2)}}
        \right)^2
        \right)^{\frac12},
        \label{eq:CN}
    \end{equation}
    and
    \begin{equation}
        W_{i,j}^{(N+2)}:=\begin{cases}
            (1-x_i^2)^{N+1}\prod_{1\leq k\leq i-1}(1-x_k^2)^{-1}, & j=i,\\
            (1-x_j^2)^{N}\prod_{1\leq k\leq j-1,k\neq i}(1-x_k^2)^{-1}, & j\neq i.
        \end{cases}
    \end{equation}
    Then achieving approximation accuracy $\epsilon$ in the $L^2(\Omega)$ norm only requires
    \begin{equation}
        N=O\!\left(\frac{\log(C_d\kappa(A)^{1/2}d^{-1/2}/\epsilon)}{\log\log(C_d\kappa(A)^{1/2}d^{-1/2}/\epsilon)}\right),
        \label{eq:main-N-eps}
    \end{equation}

    \item The condition number of the coefficient matrix satisfies
    \begin{equation}
        \kappa(L)
        \leq
        \kappa(A)\kappa(S)
        =
        O(\kappa(A)N^4).
    \end{equation}

    \item The sparsity of the coefficient matrix satisfies
    \begin{equation}
        s(L)=O(d^2N^2).
    \end{equation}
    In particular, if $A$ is diagonal, then
    \begin{equation}
        s(L)=O(dN).
    \end{equation}
\end{enumerate}
\end{lemma}

\begin{proof}
The first claim follows directly from \cor{galerkin-error-analytic} in \app{error estimation}. More precisely, that corollary gives
\begin{equation}
\|u_N-u\|_{L^2(\Omega)}
\lesssim
C_N\kappa(A)^{\frac12}d^{-\frac12}
\left(\frac{\sqrt{e/2}}{N+1}\right)^{N+\frac54},
\end{equation}
with $C_N$ defined above. Solving this estimate for $N$ yields \eqref{eq:main-N-eps}.

For the second claim, let $\boldsymbol{c}$ be an arbitrary coefficient vector and define
\begin{equation}
v_N=\sum_{\boldsymbol{k}\in[N]_0^d}c_{\boldsymbol{k}}\Phi_{\boldsymbol{k}}.
\end{equation}
By the construction of the Galerkin matrix in \sec{galerkin method}, we have
\begin{equation}
L\boldsymbol{c}=\sum_{\boldsymbol{l}\in[N]_0^d}a(v_N,\Phi_{\boldsymbol{l}})\ket{\boldsymbol{l}}.
\end{equation}
Therefore,
\begin{equation}
\boldsymbol{c}^{\mathrm T}L\boldsymbol{c}=\sum_{\boldsymbol{k}\in[N]_0^d}c_{\boldsymbol{k}}\,a(v_N,\Phi_{\boldsymbol{k}})=a(v_N,v_N).
\end{equation}
Using \eqref{eq:bilinear-form}, we obtain
\begin{equation}
\boldsymbol{c}^{\mathrm T}L\boldsymbol{c}=
\int_{\Omega}(\nabla v_N)^{\mathrm T}(A\nabla v_N)\,dx.
\end{equation}
Moreover, since $A$ is symmetric positive definite, we have
\begin{equation}\label{eq:c'Lc}
\lambda_{\min}(A)\int_{\Omega}|\nabla v_N|^2\,dx
\le
\boldsymbol{c}^\mathrm{T}L\boldsymbol{c}
\le
\lambda_{\max}(A)\int_{\Omega}|\nabla v_N|^2\,dx.
\end{equation}
Let
\begin{equation}
L_0:=\sum_{p=1}^d I^{\otimes(p-1)}\otimes S\otimes I^{\otimes(d-p)},
\end{equation}
which is precisely the coefficient matrix corresponding to the special case $A=I_d$. Then
\begin{equation}
\boldsymbol{c}^{\mathrm T}L_0\boldsymbol{c}
=
\int_{\Omega}|\nabla v_N|^2\,dx.
\end{equation}
Hence
\begin{equation}
\lambda_{\min}(A)\,\boldsymbol{c}^{\mathrm T}L_0\boldsymbol{c}
\le
\boldsymbol{c}^{\mathrm T}L\boldsymbol{c}
\le
\lambda_{\max}(A)\,\boldsymbol{c}^{\mathrm T}L_0\boldsymbol{c}.
\end{equation}
By \lem{S},
\begin{equation}
\lambda_{\min}(L_0)=d\,\lambda_{\min}(S),
\qquad
\lambda_{\max}(L_0)=d\,\lambda_{\max}(S).
\end{equation}
Therefore,
\begin{equation}
d\,\lambda_{\min}(A)\lambda_{\min}(S)\,\|\boldsymbol{c}\|^2
\le
\boldsymbol{c}^{\mathrm T}L\boldsymbol{c}
\le
d\,\lambda_{\max}(A)\lambda_{\max}(S)\,\|\boldsymbol{c}\|^2.
\end{equation}
This implies
\begin{equation}\label{eq:lambda of L}
\lambda_{\min}(L)\ge d\,\lambda_{\min}(A)\lambda_{\min}(S),
\qquad
\lambda_{\max}(L)\le d\,\lambda_{\max}(A)\lambda_{\max}(S).
\end{equation}
Consequently,
\begin{equation}
\kappa(L)
=
\frac{\lambda_{\max}(L)}{\lambda_{\min}(L)}
\le
\frac{d\,\lambda_{\max}(A)\lambda_{\max}(S)}
{d\,\lambda_{\min}(A)\lambda_{\min}(S)}
=
\kappa(A)\kappa(S).
\end{equation}

For the third claim, each term in $L_1$ has at most $N$ nonzero entries in each row and column, and there are $d$ such terms. Thus $s(L_1)\le dN$. Similarly, each term in $L_2$ has at most $N^2$ nonzero entries in each row and column, and there are at most $(d^2-d)/2$ such terms. Thus $s(L_2)\le (d^2-d)N^2/2$. Therefore,
\begin{equation}
s(L)\le s(L_1)+s(L_2)\le dN+(d^2-d)N^2/2=O(d^2N^2).
\end{equation}
If $A$ is diagonal, then $L_2=0$, and hence $s(L)\leq dN$. This completes the proof.
\end{proof}

\begin{remark}[Extension to more general boundary conditions]
\label{rem:general-boundary}

The favorable sparsity and conditioning mechanisms established above are not specific to the homogeneous Dirichlet condition. The essential ingredient is the availability of an $L^2(-1,1)$-orthonormal polynomial basis adapted to the prescribed one-dimensional boundary conditions. For any such basis, the mass matrix remains $M=I$, and the bulk terms of the Galerkin matrix retain an analogous tensor-product representation involving only one-dimensional operator matrices. Consequently, even if these one-dimensional matrices are dense, the sparsity of the multidimensional matrix remains polynomial in $N$ and $d$.

Moreover, provided that the corresponding variational problem is coercive and the polynomial approximation space satisfies appropriate inverse estimates, the same energy argument as above yields a polynomial condition-number bound. If, in addition, the approximation space retains spectral approximation properties, the main mechanism of the present spectral-Galerkin method extends to more general boundary conditions.
\end{remark}

\subsection{Rotated Jacobi-based bases}
\label{sec:rotated-basis}

We briefly discuss a basis rotation technique that can further reduce the sparsity of the coefficient matrix. This rotation does not change the approximation space, the approximation accuracy, or the condition number. Let
\begin{equation}
    U=\sum_{i,k\in[N]_0}U_{ik}\ket{i}\bra{k}\in\mathbb{R}^{N\times N}
\end{equation}
be an orthogonal matrix, and define a new one-dimensional basis by
\begin{equation}
\tilde\phi_k(x):=\sum_{i\in[N]_0}U_{ik}\phi_i(x),
\qquad k\in[N]_0.
\label{eq:rotated-basis-1d}
\end{equation}
Then
\begin{equation}
    \operatorname{span}\{\tilde\phi_k\}_{k\in[N]_0}=\operatorname{span}\{\phi_k\}_{k\in[N]_0}.
\end{equation}

Thus the approximation space is unchanged, and hence the approximation accuracy remains the same.

Let $\tilde M$, $\tilde S$, and $\tilde T$ denote the mass matrix, the stiffness matrix, and the mixed derivative matrix associated with the rotated basis. A direct computation gives
\begin{equation}
\tilde M=U^{\mathrm T}MU,
\qquad
\tilde S=U^{\mathrm T}SU,
\qquad
\tilde T=U^{\mathrm T}TU.
\label{eq:rotated-MST}
\end{equation}
Since the Jacobi-based basis \eqref{eq:new-basis-1d} satisfies $M=I$, we still have $\tilde M=I$. For the corresponding $d$-dimensional tensor-product basis, using the same rotation in each coordinate gives
\begin{equation}
\tilde L=(U^{\otimes d})^{\mathrm T}L\,U^{\otimes d}.
\label{eq:rotated-L}
\end{equation}
Therefore, $\tilde L$ is orthogonally similar to $L$, and hence
\begin{equation}
    \kappa(\tilde L)=\kappa(L).
\end{equation}
Thus the rotation does not affect the condition number.

We now consider its effect on sparsity. Since $T$ is real skew-symmetric, by the canonical form of real skew-symmetric matrices \cite{HJ13}, there exists an orthogonal matrix $U_T$ such that
\begin{equation}
U_T^{\mathrm T}TU_T=\Lambda_T,
\end{equation}
where
\begin{equation}
\Lambda_T
=
\operatorname{diag}\left(
\begin{bmatrix}
0 & \lambda_1\\
-\lambda_1 & 0
\end{bmatrix},
\ldots,
\begin{bmatrix}
0 & \lambda_{N/2}\\
-\lambda_{N/2} & 0
\end{bmatrix}
\right),
\qquad \lambda_j>0.
\label{eq:T-block-diagonal}
\end{equation}
Taking $U=U_T$, we have $\tilde T=\Lambda_T$, and the rotated coefficient matrix can be written as
\begin{equation}
    \tilde L=\tilde L_1+\tilde L_2,
\end{equation}
where
\begin{align}
\tilde L_1
&=
\sum_{p=1}^d A_{pp}
\bigl(I^{\otimes(p-1)}\otimes \tilde S\otimes I^{\otimes(d-p)}\bigr),
\label{eq:rotated-L1}\\
\tilde L_2
&=
-2\sum_{1\leq p<q\leq d}A_{pq}
\bigl(I^{\otimes(p-1)}\otimes \Lambda_T\otimes I^{\otimes(q-p-1)}
\otimes \Lambda_T\otimes I^{\otimes(d-q)}\bigr).
\label{eq:rotated-L2}
\end{align}
For the sparsity estimate, each term in $\widetilde L_1$ has at most $N$
nonzero entries in each row and column, and there are $d$ such terms.
Thus $s(\widetilde L_1)\le dN$. Similarly, each term in $\widetilde L_2$ has at most one nonzero entry in each row and column, because $\Lambda_T$ has at most one nonzero entry in each row and column. Since there are $(d^2-d)/2$ such terms, we have $s(\widetilde L_2)\le (d^2-d)/2$. Therefore,
\begin{equation}
s(\widetilde L)
\le s(\widetilde L_1)+s(\widetilde L_2)
\le dN+(d^2-d)/2
=
O(dN+d^2).
\label{eq:rotated-sparsity-general}
\end{equation}

In the special case where $A$ is diagonal, the off-diagonal part vanishes. In this case, one may instead diagonalize the symmetric stiffness matrix $S$. Let $U_S$ be an orthogonal matrix such that
\begin{equation}
U_S^{\mathrm T}SU_S=\Lambda_S.
\end{equation}
Taking $U=U_S$, the rotated coefficient matrix becomes
\begin{equation}
\tilde L
=
\sum_{p=1}^d A_{pp}
\bigl(I^{\otimes(p-1)}\otimes \Lambda_S\otimes I^{\otimes(d-p)}\bigr).
\label{eq:rotated-L-diagonal-A}
\end{equation}
Thus $\tilde L$ is diagonal. In particular,
\begin{equation}
s(\tilde L)=1.
\label{eq:rotated-sparsity-diagonal}
\end{equation}


The preceding discussion can be summarized as the following lemma.

\begin{lemma}\label{lem:rotated-jacobi-basis-properties}
Consider an instance of \prb{pde}, namely, the $d$-dimensional second-order linear elliptic boundary value problem with homogeneous Dirichlet boundary conditions \eqref{eq:model-problem}. Let the approximation space be generated by the tensor-product basis associated with the rotated one-dimensional basis \eqref{eq:rotated-basis-1d}. Then the rotated basis has the same approximation property and the same condition number bound as in \lem{jacobi-basis-properties}. More precisely:
\begin{enumerate}
    \item If the exact solution is analytic and satisfies the uniform regularity condition \eqref{eq:Cd}, then the same approximation estimate as in item 1 of \lem{jacobi-basis-properties} holds. In particular, achieving approximation accuracy $\epsilon$ in the $L^2(\Omega)$ norm requires the same asymptotic choice of $N$ as in \eqref{eq:main-N-eps}.
    
    \item The condition number of the rotated coefficient matrix satisfies
    \begin{equation}
        \kappa(\tilde L)=\kappa(L)\leq \kappa(A)\kappa(S)=O(\kappa(A)N^4).
    \end{equation}
    
    \item If $U=U_T$ is chosen so that $U_T^{\mathrm T}TU_T=\Lambda_T$ has the block diagonal form \eqref{eq:T-block-diagonal}, then the sparsity of the rotated coefficient matrix satisfies
    \begin{equation}
        s(\tilde L)=O(dN+d^2).
    \end{equation}
    In particular, if $A$ is diagonal and $U=U_S$ is chosen so that $U_S^{\mathrm T}SU_S=\Lambda_S$, then $\tilde L$ is diagonal, and hence
    \begin{equation}
        s(\tilde L)=1.
    \end{equation}
\end{enumerate}
\end{lemma}

\subsection{Classical complexity of the spectral-Galerkin method}
\label{sec:classical complexity}

For comparison, we briefly discuss the classical complexity of the spectral-Galerkin method based on both the Jacobi-based basis \eqref{eq:new-basis-1d} and its rotated counterpart \eqref{eq:rotated-basis-1d}. Let $u$ be the exact solution of \eqref{eq:model-problem}, and let
\begin{equation}
u_N(\boldsymbol{x})=\sum_{\boldsymbol{k}\in[N]_0^d} c_{\boldsymbol{k}}\Phi_{\boldsymbol{k}}(\boldsymbol{x})
\end{equation}
be the Galerkin approximation determined by the linear system
\begin{equation}
L\boldsymbol{c}=\boldsymbol{f}.
\end{equation}
Suppose that this linear system is solved approximately by the conjugate gradient method, yielding an approximate coefficient vector
\begin{equation}
\tilde{\boldsymbol{c}}=\sum_{\boldsymbol{k}\in[N]_0^d}\tilde c_{\boldsymbol{k}}\ket{\boldsymbol{k}},
\end{equation}
and hence the numerical solution
\begin{equation}
\tilde u_N(\boldsymbol{x})=\sum_{\boldsymbol{k}\in[N]_0^d}\tilde c_{\boldsymbol{k}}\Phi_{\boldsymbol{k}}(\boldsymbol{x}).
\end{equation}
Then the total $L^2(\Omega)$ error satisfies
\begin{equation}
\|u-\tilde u_N\|_{L^2(\Omega)}
\le
\|u-u_N\|_{L^2(\Omega)}
+
\|u_N-\tilde u_N\|_{L^2(\Omega)}
=:\epsilon_{\mathrm{SG}}+\epsilon_{\mathrm{CG}},
\label{eq:classical-error-split}
\end{equation}
where $\epsilon_{\mathrm{SG}}:=\|u-u_N\|_{L^2(\Omega)}$ is the spectral-Galerkin discretization error, and
\begin{equation}
\epsilon_{\mathrm{CG}}
:=
\|u_N-\tilde u_N\|_{L^2(\Omega)}
=
\|\boldsymbol{c}-\tilde{\boldsymbol{c}}\|_2
\end{equation}
is the error of the conjugate gradient method.

Under the standard sparse-matrix cost model for conjugate gradient, solving the $N^d\times N^d$ linear system to coefficient accuracy $\epsilon_{\mathrm{CG}}$ requires
\begin{equation}
O\left(
N^d\,s(L)\,\kappa(L)^{\frac12}
\log\left(
\frac{1}{\sqrt{\lambda_{\min}(L)}\,\epsilon_{\mathrm{CG}}}
\right)
\right)
\label{eq:cg-cost}
\end{equation}
arithmetic operations. Applying \lem{jacobi-basis-properties} and \lem{rotated-jacobi-basis-properties}, and choosing $\epsilon_{\mathrm{SG}},\epsilon_{\mathrm{CG}}=O(\epsilon)$, we obtain the following classical complexity estimates. Here, for simplicity, we suppress the $\log\log$ denominator in \eqref{eq:main-N-eps}, lower-order logarithmic factors coming from the CG tolerance, and inessential factors inside the logarithm, and take
\begin{equation}
    N=O\left(\log\left(C_d\kappa(A)/\epsilon\right)\right).
\end{equation}

\begin{theorem}\label{thm:classical-complexity}
Consider an instance of \prb{pde}, namely, the $d$-dimensional second-order linear elliptic boundary value problem with homogeneous Dirichlet boundary conditions \eqref{eq:model-problem}. Assume that the exact solution is analytic and satisfies the uniform regularity condition \eqref{eq:Cd}. The classical complexity of achieving overall $L^2(\Omega)$ accuracy $\epsilon$ by the conjugate gradient method satisfies the following bounds.

\begin{enumerate}
    \item For the Jacobi-based basis \eqref{eq:new-basis-1d}, the complexity is
    \begin{equation}
    O\left(
    \kappa(A)^{\frac12}\,
    d^2\,
    \left(\log(C_d\kappa(A)/\epsilon)\right)^{d+5}
    \right).
    \label{eq:classical-complexity-original}
    \end{equation}
    If $A$ is diagonal, then the bound improves to
    \begin{equation}
    O\left(
    \kappa(A)^{\frac12}\,
    d\,
    \left(\log(C_d\kappa(A)/\epsilon)\right)^{d+4}
    \right).
    \label{eq:classical-complexity-original-diagonal}
    \end{equation}

    \item For the rotated Jacobi-based basis \eqref{eq:rotated-basis-1d}, the complexity is
    \begin{equation}
    O\left(
    \kappa(A)^{\frac12}
    \left(
    d\log(C_d\kappa(A)/\epsilon)
    +
    d^2
    \right)\log(C_d\kappa(A)/\epsilon)^{d+3}
    \right).
    \label{eq:classical-complexity-rotated}
    \end{equation}
    If $A$ is diagonal, then the rotated coefficient matrix is diagonal, and the CG-based complexity becomes
    \begin{equation}
    O\left(
    \kappa(A)^{\frac12}
    \left(\log(C_d\kappa(A)/\epsilon)\right)^{d+3}
    \right).
    \label{eq:classical-complexity-rotated-diagonal}
    \end{equation}
\end{enumerate}

Here $C_d$ is defined in \eqref{eq:Cd}. In particular, for the Poisson equation $A=I_d$, the diagonal-case bounds apply, and the factors involving $\kappa(A)$ disappear.
\end{theorem}

\begin{remark}
For consistency, we use the conjugate gradient method throughout the classical complexity analysis. When $A$ is diagonal, however, the rotated Galerkin matrix is diagonal, so the linear system can be solved directly by componentwise division in $O(N^d)$ operations. Thus, the CG bound above is not sharp in this special case.
\end{remark}
\section{Quantum spectral-Galerkin method}
\label{sec:quantum spectral-Galerkin method}

In this section, we present a quantum algorithm for solving the spectral-Galerkin linear system \eqref{eq:linear-system}. The key idea is to use a quantum linear systems algorithm to prepare the quantum state corresponding to the Galerkin coefficient vector, and then estimate the target linear functional from that state by suitable post-processing. Throughout this section, we work in the sparse-access model of \lem{qlsa} and assume that the corresponding Galerkin matrix admits efficient sparse access.

We first describe, in \sec{state preparation}, how to prepare the input state associated with the right-hand side of \eqref{eq:linear-system} using the Grover--Rudolph state-preparation method. Next, in \sec{estimate the linear functional}, we explain how to estimate the linear functional $\langle \rho,u\rangle$ from the quantum output and analyze the resulting complexity. Finally, in \sec{comparison with classical complexity}, we compare the complexity of the quantum spectral-Galerkin method with that of its classical counterpart.

\subsection{State preparation}
\label{sec:state preparation}

We begin by discussing how to prepare the quantum state corresponding to the right-hand side of the linear system \eqref{eq:linear-system}. Our goal is to prepare the normalized state
\begin{equation}
\ket{f}
=
\frac{1}{\|\boldsymbol{f}\|}
\sum_{\boldsymbol{l}\in [N]_0^d}
\langle f,\Phi_{\boldsymbol{l}}\rangle
\ket{\boldsymbol{l}}.
\label{eq:rhs-state}
\end{equation}
To this end, we use the state-preparation method of Grover and Rudolph \cite{GR02}.

\begin{lemma}
\label{lem:grover-rudolph}
Let
\begin{equation}
\ket{b}
=
\frac{1}{\|\boldsymbol{b}\|}
\sum_{j=0}^{\mathcal{N}-1} b_j \ket{j},
\qquad
\boldsymbol{b}=(b_0,\dots,b_{\mathcal{N}-1})^\mathrm{T}\in\mathbb{R}^{\mathcal{N}},
\end{equation}
where $\mathcal{N}=2^m$. For each bit string $x\in\{0,1\}^k$, with $k=0,1,\dots,m$, define
\begin{equation}
W_k(x)
:=
\sum_{y\in\{0,1\}^{\,m-k}}
|b_{xy}|^2,
\label{eq:Wk-definition}
\end{equation}
where $xy$ denotes the concatenation of the prefix $x$ and the suffix $y$. For $k=0$, we write $W_0(\emptyset)=\|\boldsymbol{b}\|^2$. Assume that the following conditions hold:
\begin{enumerate}
    \item For every $x\in\{0,1\}^{k}$ and every $k=0,1,\dots,m$, one can efficiently compute $W_k(x)$.
    \item The sign of each coefficient $b_j$ can be efficiently determined. More precisely, there exists an efficiently implementable diagonal unitary $U_b^{\pm}$ such that for all $j\in[\mathcal{N}]_0$,
    \begin{equation}
        U_b^{\pm}\ket{j}=\sgn(b_j)\ket{j},
        \qquad
        \sgn(x)=
        \begin{cases}
        1, & x\ge 0,\\
        -1, & x<0.
        \end{cases}
    \end{equation}
\end{enumerate}
Then $\ket{b}$ can be prepared using $O(m)$ state-preparation stages.
\end{lemma}

By \lem{grover-rudolph}, the efficient preparation of $\ket{f}$ requires the following assumptions: 

\begin{enumerate}
    \item One can efficiently compute arbitrary partial sums of the squared magnitudes of the coefficients of $f$, namely,
    \begin{equation}
        S(\boldsymbol{a},\boldsymbol{b})
        :=
        \sum_{\boldsymbol{a}\lesssim \boldsymbol{l}\lesssim \boldsymbol{b}}
        \bigl|\langle f,\Phi_{\boldsymbol{l}}\rangle\bigr|^2,
        \qquad
        \forall\,\boldsymbol{a},\boldsymbol{b}\in [N]_0^d,
        \label{eq:partial-sum-assumption}
    \end{equation}
    where $\boldsymbol{a}\lesssim \boldsymbol{b}$ means that the integer encoded by $\boldsymbol{a}$ is no greater than that encoded by $\boldsymbol{b}$.

    \item The sign of each coefficient can be efficiently determined. More precisely, there exists an efficiently implementable diagonal unitary $U_f^{\pm}$ such that
    \begin{equation}
        U_f^{\pm}\ket{\boldsymbol{l}}
        =
        \sgn\bigl(\langle f,\Phi_{\boldsymbol{l}}\rangle\bigr)\ket{\boldsymbol{l}},
        \qquad
        \forall\,\boldsymbol{l}\in [N]_0^d.
        \label{eq:sign-oracle-f}
    \end{equation}
\end{enumerate}

As in \cite{MP16}, for a general input function $f$, the first state-preparation condition is challenging to satisfy efficiently. In particular, even if each individual coefficient can be computed efficiently, evaluating the corresponding partial sums may still be difficult in general. This assumption is typically justified only in special structured settings, for example when the coefficients admit a tensor-product form, or when the coefficient vector is sufficiently sparse and its nonzero entries can be efficiently located and summed.

For simplicity in the subsequent analysis, we henceforth assume that the state $\ket{f}$ can be prepared perfectly in time $\poly(\log(N^d))$. Under this state-preparation assumption and the sparse-access assumption stated above, \lem{qlsa}, together with the spectral-convergence, sparsity, and condition-number bounds established in \sec{jacobi-spectral-Galerkin method}, implies that an $O(\epsilon)$-approximation to the normalized Galerkin coefficient state $\ket{c}$ can be prepared with complexity
\begin{equation}
    \poly(d,\log(1/\epsilon)),
\end{equation}
where we suppress the dependence on $\kappa(A)$, $C_d$, and lower-order logarithmic factors.

\subsection{Estimating the linear functional and its complexity}
\label{sec:estimate the linear functional}

We now explain how to estimate $\langle \rho,u\rangle$ for a given $\rho\in L^2(\Omega)$. Let
\begin{equation}
\ket{c}
=
\frac{1}{\|\boldsymbol{c}\|}
\sum_{\boldsymbol{k}\in [N]_0^d}
c_{\boldsymbol{k}}\ket{\boldsymbol{k}},
\qquad
\ket{\rho}
=
\frac{1}{\alpha}
\sum_{\boldsymbol{k}\in [N]_0^d}
\langle \rho,\Phi_{\boldsymbol{k}}\rangle \ket{\boldsymbol{k}},
\label{eq:state-c-rho}
\end{equation}
where $\alpha:=\bigl(\sum_{\boldsymbol{k}\in [N]_0^d}|\langle \rho,\Phi_{\boldsymbol{k}}\rangle|^2\bigr)^{1/2}$. The state $\ket{\rho}$ can be prepared by the same method as in \sec{state preparation}, and we assume that it can be prepared perfectly in time $\poly(\log(N^d))$. Then
\begin{equation}
\langle \rho,u\rangle\approx\langle \rho,u_N\rangle=\sum_{\boldsymbol{k}\in[N]_0^d}c_{\boldsymbol{k}}\langle\rho,\Phi_{\boldsymbol{k}}\rangle=\alpha \|\boldsymbol{c}\| \langle \rho | c \rangle\approx\alpha C \langle \rho | \tilde{c} \rangle\approx\alpha C R,
\label{eq:final-estimator}
\end{equation}
where $\ket{\tilde{c}}\approx\ket{c}$ is the approximate output state produced by the QLSA, $C\approx\|\boldsymbol{c}\|$ is the output of the norm-estimation procedure, and $R\approx\langle \rho | \tilde{c} \rangle$ is obtained via the Hadamard test \cite{AJL09}. The following theorem gives the running time for estimating the linear functional $\langle \rho,u\rangle$ via the estimator $\alpha C R$.

\begin{theorem}
\label{thm:linear-functional-estimation}
Consider an instance of \prb{pde}, namely, the $d$-dimensional second-order linear elliptic boundary value problem with homogeneous Dirichlet boundary conditions \eqref{eq:model-problem}. Assume that the exact solution is analytic and satisfies the uniform regularity condition \eqref{eq:Cd}. Then there exists a quantum algorithm that outputs an estimate of the linear functional $\langle \rho,u\rangle$ with additive error at most $\epsilon\|\rho\|_{L^2(\Omega)}$, together with a flag indicating success, and succeeds with probability bounded below by a positive constant. Under the standing assumptions that the states $\ket{f}$ and $\ket{\rho}$ can be prepared perfectly in time $\poly(\log(N^d))$, the running time satisfies the following bounds.
\begin{enumerate}
    \item For the Jacobi-based basis \eqref{eq:new-basis-1d}, the running time is
    \begin{equation}
    O\!\left(
    \frac{d^2\kappa(A)\|u\|}{\epsilon}
    \poly\!\left(\log(C_d\kappa(A)/\epsilon)\right)
    \poly\!\left(
    d\log\left(C_d d\kappa(A)\|u\|/\epsilon\right)
    \right)
    \right).
    \label{eq:linear-functional-runtime-original}
    \end{equation}
    If $A$ is diagonal, then the bound improves to
    \begin{equation}
    O\!\left(
    \frac{d\kappa(A)\|u\|}{\epsilon}
    \poly\!\left(\log(C_d\kappa(A)/\epsilon)\right)
    \poly\!\left(
    d\log\left(C_d d\kappa(A)\|u\|/\epsilon\right)
    \right)
    \right).
    \label{eq:linear-functional-runtime-original-diagonal}
    \end{equation}

    \item For the rotated Jacobi-based basis \eqref{eq:rotated-basis-1d}, the running time is
    \begin{equation}
    O\!\left(
    \frac{\bigl(d\log(C_d\kappa(A)/\epsilon)+d^2\bigr)\kappa(A)\|u\|}{\epsilon}
    \poly\!\left(\log(C_d\kappa(A)/\epsilon)\right)
    \poly\!\left(
    d\log\left(C_d d\kappa(A)\|u\|/\epsilon\right)
    \right)
    \right).
    \label{eq:linear-functional-runtime-rotated}
    \end{equation}
    If $A$ is diagonal, then, with the choice of rotation described in \lem{rotated-jacobi-basis-properties}, the bound improves to
    \begin{equation}
    O\!\left(
    \frac{\kappa(A)\|u\|}{\epsilon}
    \poly\!\left(\log(C_d\kappa(A)/\epsilon)\right)
    \poly\!\left(
    d\log\left(C_d \kappa(A)\|u\|/\epsilon\right)
    \right)
    \right).
    \label{eq:linear-functional-runtime-rotated-diagonal}
    \end{equation}
\end{enumerate}

Here $\|u\|=\|u\|_{L^2(\Omega)}$, $C_d$ is defined in \eqref{eq:Cd}. In particular, for the Poisson equation $A=I_d$, the diagonal-case bounds apply, and the factors involving $\kappa(A)$ disappear.
\end{theorem}

\begin{proof}
For the approximation in \eqref{eq:final-estimator}, we write
\begin{enumerate}
    \item $\langle \rho,u_N\rangle=\langle \rho,u\rangle+\epsilon'_{\mathrm{SG}}$, where $|\epsilon'_{\mathrm{SG}}|\le \|\rho\|_{L^2(\Omega)}\,\epsilon_{\mathrm{SG}}$ and $\epsilon_{\mathrm{SG}}:=\|u-u_N\|_{L^2(\Omega)}$.

    \item $C=\|\boldsymbol{c}\|+\epsilon_C$.

    \item $\epsilon_{\mathrm{QLSA}}:=\|\ket{\tilde{c}}-\ket{c}\|$.

    \item $R=\langle \rho|\tilde{c}\rangle+\epsilon_H$.
\end{enumerate}
Then
\begin{equation}
R=\langle \rho|c\rangle+\langle \rho|(\ket{\tilde{c}}-\ket{c})+\epsilon_H =\frac{1}{\alpha\|\boldsymbol{c}\|}\langle \rho,u_N\rangle+\epsilon'_{\mathrm{QLSA}}+\epsilon_H,
\label{eq:R-expansion}
\end{equation}
for some $\epsilon'_{\mathrm{QLSA}}$ satisfying $|\epsilon'_{\mathrm{QLSA}}|\le \epsilon_{\mathrm{QLSA}}$ by the Cauchy--Schwarz inequality. Hence
\begin{align}
\alpha C R
&=
\alpha\bigl(\|\boldsymbol{c}\|+\epsilon_C\bigr)
\left(
\frac{1}{\alpha\|\boldsymbol{c}\|}
\langle \rho,u_N\rangle
+
\epsilon'_{\mathrm{QLSA}}
+
\epsilon_H
\right)\\
&=
\bigl(\langle \rho,u\rangle+\epsilon'_{\mathrm{SG}}\bigr)
\left(
1+\frac{\epsilon_C}{\|\boldsymbol{c}\|}
\right)
+
\alpha\bigl(\|\boldsymbol{c}\|+\epsilon_C\bigr)
\bigl(\epsilon'_{\mathrm{QLSA}}+\epsilon_H\bigr).
\label{eq:alphaCR-expand}
\end{align}
Therefore,
\begin{equation}
\alpha C R-\langle \rho,u\rangle=\epsilon'_{\mathrm{SG}}\left(1+\frac{\epsilon_C}{\|\boldsymbol{c}\|}\right)+\epsilon_C'\frac{\|\rho\|_{L^2(\Omega)}\|u\|_{L^2(\Omega)}}{\|\boldsymbol{c}\|}+\alpha\bigl(\|\boldsymbol{c}\|+\epsilon_C\bigr)\bigl(\epsilon'_{\mathrm{QLSA}}+\epsilon_H\bigr),
\label{eq:total-error-decomposition}
\end{equation}
for some $\epsilon_C'$ satisfying $|\epsilon_C'|\le \epsilon_C$ by the Cauchy--Schwarz inequality. It suffices to make each term on the right-hand side of \eqref{eq:total-error-decomposition} of order $O(\epsilon\|\rho\|_{L^2(\Omega)})$. Thus we choose
\begin{align}
\epsilon_{\mathrm{SG}}&=O(\epsilon),\\
\epsilon_C&=O\!\left(\min\left\{\|\boldsymbol{c}\|,\frac{\epsilon\|\boldsymbol{c}\|}{\|u\|_{L^2(\Omega)}}\right\}\right),\\
\epsilon_{\mathrm{QLSA}},\epsilon_H&=
O\!\left(
\frac{\epsilon\|\rho\|_{L^2(\Omega)}}{\alpha\|\boldsymbol{c}\|}
\right).
\end{align}
For simplicity, assume $\epsilon\le \|u\|_{L^2(\Omega)}$. Then
\begin{equation}
\epsilon_C=O\!\left(\frac{\epsilon\|\boldsymbol{c}\|}{\|u\|_{L^2(\Omega)}}\right).
\label{eq:epsilonC-choice}
\end{equation}
Since $\{\Phi_{\boldsymbol{k}}\}$ is an orthonormal basis, we have $\alpha\le \|\rho\|_{L^2(\Omega)}$, and 
\begin{equation}
\|\boldsymbol{c}\|=\|u_N\|_{L^2(\Omega)}\le \|u\|_{L^2(\Omega)}+\epsilon_{\mathrm{SG}}=O(\|u\|_{L^2(\Omega)}).
\label{eq:c-norm-upper}
\end{equation}
Therefore, it suffices to choose
\begin{equation}
\epsilon_{\mathrm{QLSA}},\epsilon_H
=
O\!\left(
\epsilon\|u\|_{L^2(\Omega)}^{-1}
\right).
\label{eq:eps-qlsa-hadamard-choice}
\end{equation}

Since $\epsilon_{\mathrm{SG}}=O(\epsilon)$, both \lem{jacobi-basis-properties} and \lem{rotated-jacobi-basis-properties} lead to the same asymptotic choice
\begin{equation}
N=O\left(\log\left(C_d\kappa(A)/\epsilon\right)\right).
\end{equation}
Here, we suppress the $\log\log$ denominator and inessential lower-order logarithmic factors.

Next we estimate the cost of approximating $\|\boldsymbol{c}\|$. We use the norm-estimation algorithm from \cite{Dalzell24}. After accounting for the known normalization factors in the block encoding of $L$ and the preparation of the normalized right-hand side state $\ket{f}$, estimating $\|\boldsymbol{c}\|$ to relative error $\delta$ has total running time
\begin{equation}
O\!\left(
\frac{s\kappa}{\delta}
\poly\!\left(\log\left(N^ds\kappa/\delta\right)\right)
\right),
\label{eq:norm-estimation-cost}
\end{equation}
where $s:=s(L)$ and $\kappa:=\kappa(L)$. These normalization factors only rescale the solution norm by a known constant and therefore do not affect the relative error $\delta$ or the asymptotic complexity. Here the complexity in \cite{Dalzell24} is stated in terms of queries to a block encoding, and we absorb the block-encoding cost into the overall complexity bound. By \eqref{eq:epsilonC-choice}, we need
\begin{equation}
\delta
=
O\!\left(\frac{\epsilon}{\|u\|_{L^2(\Omega)}}\right).
\label{eq:delta-choice}
\end{equation}
Therefore, the total running time for estimating $\|\boldsymbol{c}\|$ is
\begin{equation}
O\!\left(
\frac{s\kappa\|u\|_{L^2(\Omega)}}{\epsilon}
\poly\!\left(
\log\left(N^ds\kappa\|u\|_{L^2(\Omega)}/\epsilon\right)
\right)
\right).
\label{eq:norm-estimation-final-general}
\end{equation}

We now estimate the total cost of the Hadamard-test step. By \lem{qlsa}, preparing one copy of $\ket{\tilde{c}}$ requires time
\begin{equation}
O\!\left(
s\kappa\,
\poly\!\left(
\log\left(N^ds\kappa/\epsilon_{\mathrm{QLSA}}\right)
\right)
\right).
\label{eq:one-qlsa-cost}
\end{equation}
Combining the Hadamard-test circuit with amplitude estimation \cite{BHMT02}, the quantity $\langle \rho|\tilde{c}\rangle$ can be estimated to additive accuracy $\epsilon_H$ using $O(1/\epsilon_H)$ controlled uses of the state-preparation routines for $\ket{\tilde{c}}$ and $\ket{\rho}$. Therefore, the total running time of the Hadamard-test step is
\begin{equation}
O\!\left(
\frac{s\kappa\|u\|_{L^2(\Omega)}}{\epsilon}
\poly\!\left(
\log\left(N^ds\kappa\|u\|_{L^2(\Omega)}/\epsilon\right)
\right)
\right).
\label{eq:hadamard-final-general}
\end{equation}
This has the same asymptotic order as the total running time in \eqref{eq:norm-estimation-final-general}.

Finally, we substitute the bounds on the sparsity and condition number of the coefficient matrix into the complexity estimates above.

For the Jacobi-based basis \eqref{eq:new-basis-1d}, \lem{jacobi-basis-properties} gives
\begin{equation}
\kappa(L)=O(\kappa(A)N^4),
\qquad
s(L)=O(d^2N^2),
\end{equation}
and, if $A$ is diagonal, then $s(L)=O(dN)$. Substituting these bounds, together with $N=O(\log(C_d\kappa(A)/\epsilon))$, into \eqref{eq:norm-estimation-final-general} and \eqref{eq:hadamard-final-general}, we obtain \eqref{eq:linear-functional-runtime-original} and \eqref{eq:linear-functional-runtime-original-diagonal}.

For the rotated Jacobi-based basis \eqref{eq:rotated-basis-1d}, \lem{rotated-jacobi-basis-properties} gives
\begin{equation}
\kappa(\tilde L)=O(\kappa(A)N^4),
\qquad
s(\tilde L)=O(dN+d^2),
\end{equation}
and, if $A$ is diagonal, then $s(\tilde L)=1$. Substituting these bounds, together with $N=O(\log(C_d\kappa(A)/\epsilon))$, into \eqref{eq:norm-estimation-final-general} and \eqref{eq:hadamard-final-general}, we obtain \eqref{eq:linear-functional-runtime-rotated} and \eqref{eq:linear-functional-runtime-rotated-diagonal}.
\end{proof}

\subsection{Comparison with classical complexity}
\label{sec:comparison with classical complexity}
To put the above results into perspective, we briefly compare the quantum complexity with that of the corresponding classical spectral-Galerkin methods. For simplicity, in this subsection we focus on the comparison for general $A$, and suppress the dependence on $C_d$, $\kappa(A)$, and lower-order logarithmic factors. The diagonal case can be analyzed similarly by using the improved bounds in \thm{classical-complexity} and \thm{linear-functional-estimation}.

We first note that this comparison depends on what is regarded as the output. If the goal is to compute the full coefficient vector $\boldsymbol{c}$, then the classical algorithm outputs $\boldsymbol{c}$ itself, whereas the quantum algorithm based on the QLSA only prepares the quantum state $\ket{c}$ associated with $\boldsymbol{c}$. Thus, these two tasks are not directly equivalent. Nevertheless, the comparison is still informative at the level of solving the underlying linear system.

For the \textbf{classical algorithm}, \thm{classical-complexity} gives the following bounds. For the Jacobi-based basis \eqref{eq:new-basis-1d}, the complexity of computing the full coefficient vector is of order
\begin{equation}
O\!\left(
d^2(\log(1/\epsilon))^{d+5}
\right),
\end{equation}
whereas for the rotated Jacobi-based basis \eqref{eq:rotated-basis-1d}, it becomes
\begin{equation}
O\!\left(
\left(d\log(1/\epsilon)+d^2\right)(\log(1/\epsilon))^{d+3}
\right).
\end{equation}
In both cases, the classical complexity still grows exponentially with the dimension $d$. By contrast, for the \textbf{quantum algorithm}, the cost of preparing the corresponding solution state is determined by \eqref{eq:one-qlsa-cost}. In particular, for the Jacobi-based basis \eqref{eq:new-basis-1d}, this cost is of order
\begin{equation}
O\!\left(d^2\,\poly\!\left(d,\log(1/\epsilon)\right)\right),
\end{equation}
while for the rotated Jacobi-based basis \eqref{eq:rotated-basis-1d}, it becomes
\begin{equation}
O\!\left((d\log(1/\epsilon)+d^2)\,\poly\!\left(d,\log(1/\epsilon)\right)\right).
\end{equation}
Therefore, at the level of solving for the Galerkin coefficients, the quantum algorithm avoids the exponential growth in the dimension $d$ that appears in the classical complexity.

However, the quantum state $\ket{c}$ only provides an encoded representation of the solution, and in general one cannot efficiently extract all coefficients $c_{\boldsymbol{k}}$ from it. For this reason, we use the estimation of linear functionals as the complexity measure for the quantum algorithm. In this setting, \thm{linear-functional-estimation} shows that the complexity is
\begin{equation}
O\!\left(
\frac{d^2}{\epsilon}\,\poly\!\left(d,\log(1/\epsilon)\right)
\right)
\end{equation}
for the Jacobi-based basis, and
\begin{equation}
O\!\left(
\frac{d\log(1/\epsilon)+d^2}{\epsilon}\,\poly\!\left(d,\log(1/\epsilon)\right)
\right)
\end{equation}
for the rotated Jacobi-based basis. Under this output model, the advantage of the quantum spectral-Galerkin method over its classical counterpart is less pronounced in low dimensions, since the dominant cost comes from norm estimation and Hadamard-test-based estimation, which introduce the factor $1/\epsilon$. Nevertheless, for high-dimensional problems, the quantum spectral-Galerkin method still offers a significant advantage, since its dependence on $d$ remains polynomial rather than exponential.
\section{Conclusion and open problems}
\label{sec:conclusion}

In this paper, we presented a quantum spectral-Galerkin method for the $d$-dimensional second-order linear elliptic boundary value problem \eqref{eq:model-problem} with homogeneous Dirichlet boundary conditions. By combining Jacobi-based tensor-product bases, including a rotated variant, with a sparse-input QLSA, we obtained Galerkin coefficient matrices whose sparsity and condition number grow only polynomially with the discretization parameter and the dimension. As a result, for analytic solutions satisfying the uniform regularity condition \eqref{eq:Cd}, the complexity of preparing the quantum state associated with the Galerkin coefficient vector is $\poly(d,\log(1/\epsilon))$, improving the precision dependence over fixed-order finite-difference and finite-element discretizations, where we suppress the additional dependence on $\kappa(A)$, $C_d$, and lower-order logarithmic factors. When the output is taken to be a linear functional of the solution, the overall complexity acquires an additional factor $1/\epsilon$ through norm estimation and Hadamard-test-based estimation.

Our analysis leaves several natural questions for future work.
\begin{enumerate}
    \item \textbf{State preparation for general data.}
    State preparation remains a major bottleneck. As in the existing literature on quantum finite element methods \cite{MP16,QRFEM,QAFEM}, we assume that the relevant input states $\ket{f}$ and $\ket{\rho}$ can be prepared efficiently. For general right-hand sides and observables, however, it remains unclear when this assumption is justified. A better understanding of the structural conditions that enable efficient state preparation for PDE-related data is therefore an important direction for future work.

    \item \textbf{Variable-coefficient problems.}
    It is natural to ask whether the present framework can be extended to variable-coefficient equations. In this setting, the simple tensor-product structure exploited in the present analysis is generally lost, potentially leading to denser matrices and less favorable condition numbers. A possible approach is to use the constant-coefficient part as a preconditioner, although it remains unclear whether the resulting preconditioned system admits an efficient sparse representation and quantum implementation.

    \item \textbf{Weaker regularity assumptions and broader PDE classes.}
    The complexity of our algorithm depends on the regularity quantity $C_d$ in \eqref{eq:Cd}, which encodes the size of certain high-order derivatives of the solution. It is natural to ask whether similar quantum complexity guarantees can be obtained under weaker regularity assumptions or without explicit dependence on such high-order derivative information. More broadly, one may ask whether the ideas developed here can be extended to other classes of PDEs, such as convection-diffusion equations, time-dependent equations, higher-order elliptic equations, or nonlinear problems.
\end{enumerate}
These questions touch on fundamental structural and analytical constraints that determine whether a spectral discretization is compatible with efficient quantum linear algebra. Answering them will help clarify the ultimate scope of quantum spectral methods for partial differential equations.
\section*{Acknowledgments}

The authors acknowledge support from the National Natural Science Foundation of China under Grant No.~12471341.


\providecommand{\bysame}{\leavevmode\hbox to3em{\hrulefill}\thinspace}

\newpage

\appendix

\section{Derivation of the Galerkin system}
\label{app:derivation}

In this appendix, we derive the tensor-product representation
\eqref{eq:L1-final}--\eqref{eq:L2-final} of the Galerkin coefficient matrix.
Starting from \eqref{eq:discrete-entrywise}, we decompose the discrete operator according to the diagonal terms $p=q$ and the off-diagonal terms $p\neq q$.

For the diagonal contribution, define
\begin{equation}
L_1\boldsymbol{c}:=
\sum_{\boldsymbol{l}\in[N]_0^d}
\left(
\sum_{p=1}^d A_{pp}\sum_{\boldsymbol{k}\in[N]_0^d}c_{\boldsymbol{k}}
\bigl\langle \partial_{x_p}\Phi_{\boldsymbol{k}},\partial_{x_p}\Phi_{\boldsymbol{l}}\bigr\rangle
\right)\ket{\boldsymbol{l}}.
\label{eq:L1-action}
\end{equation}
By the tensor-product structure of the basis functions, we obtain
\begin{align}
L_1\boldsymbol{c}
&=
\sum_{p=1}^d A_{pp}
\sum_{\boldsymbol{l},\boldsymbol{k}\in[N]_0^d}
c_{\boldsymbol{k}}
\langle\phi_{k_p}',\phi_{l_p}'\rangle
\prod_{j\in[d]\setminus\{p\}}
\langle\phi_{k_j},\phi_{l_j}\rangle
\ket{\boldsymbol{l}}
\notag\\
&=
\sum_{p=1}^d A_{pp}
\sum_{\boldsymbol{l},\boldsymbol{k}\in[N]_0^d}
c_{\boldsymbol{k}}
S_{l_p,k_p}
\prod_{j\in[d]\setminus\{p\}}
M_{l_j,k_j}
\ket{\boldsymbol{l}}
\notag\\
&=
\sum_{p=1}^d A_{pp}
\sum_{\boldsymbol{k}\in[N]_0^d}
c_{\boldsymbol{k}}
\left(
\bigotimes_{j=1}^{p-1}\sum_{l_j\in[N]_0}M_{l_j,k_j}\ket{l_j}
\right)
\otimes
\left(
\sum_{l_p\in[N]_0}S_{l_p,k_p}\ket{l_p}
\right)
\otimes
\left(
\bigotimes_{j=p+1}^{d}\sum_{l_j\in[N]_0}M_{l_j,k_j}\ket{l_j}
\right)
\notag\\
&=
\sum_{p=1}^d A_{pp}
\sum_{\boldsymbol{k}\in[N]_0^d}
c_{\boldsymbol{k}}
\left(
\bigotimes_{j=1}^{p-1}M\ket{k_j}
\right)
\otimes
\left(
S\ket{k_p}
\right)
\otimes
\left(
\bigotimes_{j=p+1}^{d}M\ket{k_j}
\right)
\notag\\
&=
\sum_{p=1}^d A_{pp}
\bigl(M^{\otimes(p-1)}\otimes S\otimes M^{\otimes(d-p)}\bigr)\boldsymbol{c}.
\label{eq:L1-tensor}
\end{align}
For the off-diagonal contribution, define
\begin{equation}
L_2\boldsymbol{c}:=
\sum_{\boldsymbol{l}\in[N]_0^d}
\left(
\sum_{p\neq q} A_{pq}\sum_{\boldsymbol{k}\in[N]_0^d}c_{\boldsymbol{k}}
\bigl\langle \partial_{x_p}\Phi_{\boldsymbol{k}},\partial_{x_q}\Phi_{\boldsymbol{l}}\bigr\rangle
\right)\ket{\boldsymbol{l}}.
\label{eq:L2-action}
\end{equation}
A similar calculation applies to the off-diagonal terms. Pairing the contributions indexed by $(p,q)$ and $(q,p)$, using the identity $\langle \phi_k,\phi_j'\rangle=-\langle \phi_k',\phi_j\rangle$, and invoking the symmetry $A_{pq}=A_{qp}$, we obtain
\begin{equation}
L_2\boldsymbol{c}
=
-2\sum_{1\leq p<q\leq d}A_{pq}
\bigl(M^{\otimes(p-1)}\otimes T\otimes M^{\otimes(q-p-1)}\otimes T\otimes M^{\otimes(d-q)}\bigr)\boldsymbol{c}.
\label{eq:L2-tensor}
\end{equation}
\section{Jacobi polynomials}
\label{app:Jacobi}

This appendix collects several basic facts about Jacobi polynomials that are used in our analysis. More details can be found in \cite{STW11}.

For parameters $\alpha,\beta>-1$, the Jacobi polynomial $P_n^{(\alpha,\beta)}(x)$ of degree $n$ can be written explicitly as
\begin{equation}
P_n^{(\alpha,\beta)}(x)
=
\frac{\Gamma(n+\alpha+1)}{n!\,\Gamma(n+\alpha+\beta+1)}
\sum_{k=0}^{n}
\binom{n}{k}
\frac{\Gamma(n+k+\alpha+\beta+1)}{\Gamma(k+\alpha+1)}
\left(\frac{x-1}{2}\right)^k .
\label{eq:jacobi-explicit}
\end{equation}
Although the explicit formula \eqref{eq:jacobi-explicit} is available, it is not the most convenient form for numerical computation. In practice, Jacobi polynomials can be evaluated efficiently by the three-term recurrence
\begin{equation}
P_{n+1}^{(\alpha,\beta)}(x)
=
\bigl(a_n^{(\alpha,\beta)}x-b_n^{(\alpha,\beta)}\bigr)P_n^{(\alpha,\beta)}(x)
-
c_n^{(\alpha,\beta)}P_{n-1}^{(\alpha,\beta)}(x),
\qquad n\ge 1,
\label{eq:jacobi-recurrence}
\end{equation}
with initial values
\begin{equation}
P_0^{(\alpha,\beta)}(x)=1,
\qquad
P_1^{(\alpha,\beta)}(x)
=
\frac{\alpha+\beta+2}{2}x+\frac{\alpha-\beta}{2},
\label{eq:jacobi-initial}
\end{equation}
and coefficients
\begin{align}
a_n^{(\alpha,\beta)}
&=
\frac{(2n+\alpha+\beta+1)(2n+\alpha+\beta+2)}
{2(n+1)(n+\alpha+\beta+1)},
\label{eq:jacobi-an}
\\
b_n^{(\alpha,\beta)}
&=
\frac{(\beta^2-\alpha^2)(2n+\alpha+\beta+1)}
{2(n+1)(n+\alpha+\beta+1)(2n+\alpha+\beta)},
\label{eq:jacobi-bn}
\\
c_n^{(\alpha,\beta)}
&=
\frac{(n+\alpha)(n+\beta)(2n+\alpha+\beta+2)}
{(n+1)(n+\alpha+\beta+1)(2n+\alpha+\beta)}.
\label{eq:jacobi-cn}
\end{align}
Therefore, once the two initial polynomials are known, all higher-order Jacobi polynomials can be generated recursively.

Jacobi polynomials form an orthogonal polynomial system on $[-1,1]$ with respect to the weight
\begin{equation}
\omega^{(\alpha,\beta)}(x)=(1-x)^\alpha(1+x)^\beta .
\label{eq:jacobi-weight}
\end{equation}
More precisely, they satisfy
\begin{equation}
\int_{-1}^{1}
P_n^{(\alpha,\beta)}(x)\,P_m^{(\alpha,\beta)}(x)\,\omega^{(\alpha,\beta)}(x)\,dx
=
\gamma_n^{(\alpha,\beta)}\delta_{nm},
\label{eq:jacobi-orthogonality}
\end{equation}
where
\begin{equation}
\gamma_n^{(\alpha,\beta)}
=
\frac{2^{\alpha+\beta+1}\Gamma(n+\alpha+1)\Gamma(n+\beta+1)}
{(2n+\alpha+\beta+1)n!\,\Gamma(n+\alpha+\beta+1)}.
\label{eq:jacobi-norm}
\end{equation}
When $\alpha=\beta=0$, $P_n^{(0,0)}(x)$ coincides with the Legendre polynomial of degree $n$. When $\alpha=\beta=-1/2$, $P_n^{(-1/2,-1/2)}(x)$ is proportional to the Chebyshev polynomial of the first kind of degree $n$.

Next, we introduce an important derivative formula. The derivative of a Jacobi polynomial is still a Jacobi polynomial with shifted parameters. More generally, for $n\ge k$,
\begin{equation}
\frac{d^k}{dx^k}P_n^{(\alpha,\beta)}(x)
=
d_{n,k}^{(\alpha,\beta)}\,
P_{n-k}^{(\alpha+k,\beta+k)}(x),
\label{eq:jacobi-derivative}
\end{equation}
where
\begin{equation}
d_{n,k}^{(\alpha,\beta)}
=
\frac{\Gamma(n+k+\alpha+\beta+1)}
{2^k\,\Gamma(n+\alpha+\beta+1)}.
\label{eq:jacobi-derivative-coef}
\end{equation}
This property is especially useful in spectral methods, because derivatives can still be evaluated through Jacobi polynomials and hence can be computed efficiently by the recurrence relation \eqref{eq:jacobi-recurrence}.

Jacobi polynomials are also complete in $L^2_{\omega^{(\alpha,\beta)}}(-1,1)$. Hence, $\{P_n^{(\alpha,\beta)}(x)\}_{n\ge 0}$ forms a complete orthogonal basis of $L^2_{\omega^{(\alpha,\beta)}}(-1,1)$. Therefore, any function $u\in L^2_{\omega^{(\alpha,\beta)}}(-1,1)$ admits a unique Jacobi series expansion
\begin{equation}
u(x)=\sum_{n=0}^{\infty}\hat{u}_n^{(\alpha,\beta)} P_n^{(\alpha,\beta)}(x),
\label{eq:jacobi-expansion}
\end{equation}
where the coefficients are given by
\begin{equation}
\hat{u}_n^{(\alpha,\beta)}
=
\frac{1}{\gamma_n^{(\alpha,\beta)}}
\int_{-1}^{1}
u(x)P_n^{(\alpha,\beta)}(x)\omega^{(\alpha,\beta)}(x)\,dx .
\label{eq:jacobi-coeff}
\end{equation}

Finally, we define the $L^2_{\omega^{(\alpha,\beta)}}$-orthogonal projection $\pi_N^{(\alpha,\beta)} : L^2_{\omega^{(\alpha,\beta)}}(-1,1)\to \mathbb{P}_N$ such that
\begin{equation}
\bigl(\pi_N^{(\alpha,\beta)}u-u,v_N\bigr)_{\omega^{(\alpha,\beta)}}=0,
\qquad
\forall\, v_N\in \mathbb{P}_N,
\label{eq:jacobi-projection-def}
\end{equation}
or equivalently,
\begin{equation}
(\pi_N^{(\alpha,\beta)}u)(x)
=
\sum_{n=0}^{N}\hat{u}_n^{(\alpha,\beta)} P_n^{(\alpha,\beta)}(x),
\label{eq:jacobi-projection-series}
\end{equation}
where $\hat{u}_n^{(\alpha,\beta)}$ is given by \eqref{eq:jacobi-coeff}. The following proposition summarizes the approximation property of this projection.

\begin{proposition}
\label{prop:jacobi-projection-estimate}
Let $\alpha,\beta>-1$ and $m\in\mathbb{N}$. Define
\begin{equation}
B_{\alpha,\beta}^m(-1,1)
:=
\left\{
u\mid\partial_x^k u\in L^2_{\omega^{(\alpha+k,\beta+k)}}(-1,1),\,\, k=0,1,\ldots,m
\right\}.
\end{equation}
For any $u\in B_{\alpha,\beta}^m(-1,1)$, the following hold:
\begin{enumerate}
    \item If $0\le l\le m\le N+1$, then
    \begin{equation}
    \left\|
    \partial_x^l\bigl(\pi_N^{(\alpha,\beta)}u-u\bigr)
    \right\|_{\omega^{(\alpha+l,\beta+l)}}
    \lesssim N^{l-m}
    \left\|
    \partial_x^m u
    \right\|_{\omega^{(\alpha+m,\beta+m)}}.
    \label{eq:jacobi-projection-estimate-1}
    \end{equation}
    
    \item If $m>N+1$ and $0\le l\le N+1$, then
    \begin{equation}
    \left\|
    \partial_x^l\bigl(\pi_N^{(\alpha,\beta)}u-u\bigr)
    \right\|_{\omega^{(\alpha+l,\beta+l)}}
    \lesssim 
    \left(\frac{\sqrt{e/2}}{N}\right)^{N-l+\frac{5}{4}}
    \left\|
    \partial_x^{N+1}u
    \right\|_{\omega^{(\alpha+N+1,\beta+N+1)}}.
    \label{eq:jacobi-projection-estimate-2}
    \end{equation}
\end{enumerate}
\end{proposition}

\begin{remark}
In the proof of \prop{jacobi-projection-estimate}, \cite{STW11} actually established the more precise estimate
\begin{equation}
\left\|
\partial_x^l\bigl(\pi_N^{(\alpha,\beta)}u-u\bigr)
\right\|_{\omega^{(\alpha+l,\beta+l)}}
\le
c\sqrt{\frac{(N-\tilde m+1)!}{(N-l+1)!}}\,(N+\tilde m)^{(l-\tilde m)/2}
\left\|
\partial_x^{\tilde m}u
\right\|_{\omega^{(\alpha+\tilde m,\beta+\tilde m)}},
\end{equation}
where $\tilde m=\min\{m,N+1\}$, and $c\approx 1$ for $N\gg 1$. \prop{jacobi-projection-estimate} follows from this estimate by Stirling's formula. 

In the subsequent analysis, we shall need the case $\tilde m=N$. Let $r:=N-l+1$. Then, by Stirling's formula,
\begin{align}
&\sqrt{\frac{(N-\tilde m+1)!}{(N-l+1)!}}(N+\tilde m)^{(l-\tilde m)/2}
=
\frac{(2N)^{(l-N)/2}}{\sqrt{(N-l+1)!}}
=
\frac{(2N)^{(1-r)/2}}{\sqrt{r!}} \\
\sim&
(2N)^{(1-r)/2}(2\pi r)^{-1/4}\left(\frac{e}{r}\right)^{r/2}
=(2N)^{1/2}(2\pi r)^{-1/4}\left(\frac{e}{2Nr}\right)^{r/2}\sim
\left(\frac{\sqrt{e/2}}{N}\right)^{N-l+\frac34}.
\end{align}
Therefore, for each fixed $l$ and $N\gg1$, if $u\in B_{\alpha,\beta}^N(-1,1)$, then
\begin{equation}
\left\|
\partial_x^l\bigl(\pi_N^{(\alpha,\beta)}u-u\bigr)
\right\|_{\omega^{(\alpha+l,\beta+l)}}
\lesssim
\left(\frac{\sqrt{e/2}}{N}\right)^{N-l+\frac34}
\left\|
\partial_x^N u
\right\|_{\omega^{(\alpha+N,\beta+N)}}.
\label{eq:jacobi-projection-estimate-N}
\end{equation}
\end{remark}

\section{Error estimation of spectral-Galerkin methods}
\label{app:error estimation}

In this section, we study the error of the spectral-Galerkin approximation. Recall that we consider the boundary value problem
\begin{equation}
-\sum_{p,q=1}^{d} A_{pq}\,\partial_{x_p}\partial_{x_q}u=f\quad \text{in } \Omega:=(-1,1)^d,
\qquad
u|_{\partial\Omega}=0,
\label{eq:elliptic-bvp-recall}
\end{equation}
where $A=(A_{pq})_{p,q=1}^d\in\mathbb{R}^{d\times d}$ is a constant symmetric positive-definite matrix. The weak solution $u\in H_0^1(\Omega)$ satisfies
\begin{equation}
a(u,v):=\int_\Omega (\nabla u)^{\mathrm T}(A\nabla v)\,\mathrm{d}x
=
(f,v)_{L^2},
\qquad
\forall\, v\in H_0^1(\Omega).
\label{eq:weak-form-recall}
\end{equation}
Consider the finite-dimensional subspace
\begin{equation}
V_N
:=
\mathrm{span}
\left\{
\Phi_{\boldsymbol{k}}(\boldsymbol{x})
=
\prod_{j=1}^d \phi_{k_j}(x_j)
\;\middle|\;
\boldsymbol{k}\in [N]_0^d
\right\},
\label{eq:VN-recall}
\end{equation}
where the one-dimensional basis functions $\{\phi_k\}_{k\in [N]_0}$ are defined in \eqref{eq:new-basis-1d}. Equivalently, 
\begin{equation}
    V_N=(\mathbb{P}_{N+1}^0)^{\otimes d},\quad \mathbb{P}_{N+1}^0:=\{u\in\mathbb{P}_{N+1}\mid u(\pm1)=0\}.
\end{equation}
The Galerkin solution $u_N\in V_N$ is defined by
\begin{equation}
a(u_N,v_N)=(f,v_N)_{L^2},
\qquad
\forall\, v_N\in V_N.
\label{eq:galerkin-recall}
\end{equation}
Our goal is to estimate the error between $u_N$ and the exact solution $u$.

Define the energy norm by $\|v\|_a:=\sqrt{a(v,v)},\forall\,v\in H_0^1(\Omega)$. By \eqref{eq:weak-form-recall} and \eqref{eq:galerkin-recall}, we have the Galerkin orthogonality
\begin{equation}
a(u_N-u,v_N)=0,
\qquad
\forall\, v_N\in V_N.
\label{eq:galerkin-orthogonality-recall}
\end{equation}
Hence $u_N$ is the best approximation to $u$ in the energy norm over $V_N$, namely,
\begin{equation}
\|u_N-u\|_a
=
\inf_{v_N\in V_N}\|v_N-u\|_a.
\label{eq:best-approximation-recall}
\end{equation}
Indeed, for any $v_N\in V_N$,
\begin{equation}
\|u_N-u\|_a^2
=a(u_N-u,u_N-u)=
a(u_N-u,v_N-u)
\le
\|u_N-u\|_a\,\|v_N-u\|_a.
\end{equation}

Moreover, by the positive definiteness of $A$,
\begin{equation}
\sqrt{\lambda_{\min}(A)}\,\|\nabla v\|_{L^2}
\le
\|v\|_a
\le
\sqrt{\lambda_{\max}(A)}\,\|\nabla v\|_{L^2},
\qquad
\forall\, v\in H_0^1(\Omega).
\label{eq:energy-gradient-equivalence-recall}
\end{equation}
Combining \eqref{eq:best-approximation-recall} and \eqref{eq:energy-gradient-equivalence-recall}, we obtain
\begin{equation}
\|\nabla(u_N-u)\|_{L^2}
\le
\sqrt{\kappa(A)}
\inf_{v_N\in V_N}\|\nabla(v_N-u)\|_{L^2}.
\label{eq:H1-best-approximation-recall}
\end{equation}
Since $u_N-u\in H_0^1(\Omega)$, the Poincar\'e inequality further gives
\begin{equation}
\|u_N-u\|_{H^1}
\lesssim
\|\nabla(u_N-u)\|_{L^2}
\le
\sqrt{\kappa(A)}
\inf_{v_N\in V_N}\|\nabla(v_N-u)\|_{L^2}.
\label{eq:H1-error-reduction-recall}
\end{equation}

Therefore, the error analysis of the Galerkin solution is reduced to the problem of constructing a suitable approximation $v_N\in V_N$ for $u$. In the following subsections, we introduce several orthogonal projection operators, because their approximation errors can be analyzed explicitly and can therefore be used to control the error of the spectral-Galerkin solution.

\subsection{Generalized Jacobi polynomials}

In this subsection, we briefly recall the generalized Jacobi polynomials needed for the error analysis of our spectral-Galerkin method. More details can be found in \cite{GSW09,GW04}.

We consider the generalized Jacobi polynomials associated with the parameters $\alpha=\beta=-1$, defined by
\begin{equation}
P_n^{(-1,-1)}(x)=(1-x^2)P_{n-2}^{(1,1)}(x),
\qquad n\ge 2.
\end{equation}
They form an orthogonal polynomial system on $[-1,1]$ with respect to the weight
\begin{equation}
\omega^{(-1,-1)}(x)=(1-x^2)^{-1}.
\end{equation}
More precisely, for $m,n\ge 2$, they satisfy
\begin{equation}
\int_{-1}^1 P_n^{(-1,-1)}(x)P_m^{(-1,-1)}(x)\omega^{(-1,-1)}(x)\,\mathrm{d}x
=
\int_{-1}^1 P_{n-2}^{(1,1)}(x)P_{m-2}^{(1,1)}(x)\omega^{(1,1)}(x)\,\mathrm{d}x
=
\gamma_{n-2}^{(1,1)}\delta_{mn}.
\end{equation}

The generalized Jacobi polynomials $\{P_n^{(-1,-1)}(x)\}_{n\ge 2}$ form a complete orthogonal basis of $L^2_{\omega^{(-1,-1)}}(-1,1)$. Therefore, any function $u\in L^2_{\omega^{(-1,-1)}}(-1,1)$ admits the unique expansion
\begin{equation}
u(x)=\sum_{n=2}^{\infty}\hat{u}_n^{(-1,-1)}P_n^{(-1,-1)}(x),
\label{eq:generalized-jacobi-expansion}
\end{equation}
where the coefficients are given by
\begin{equation}
\hat{u}_n^{(-1,-1)}
=
\frac{(u,P_n^{(-1,-1)})_{\omega^{(-1,-1)}}}
{\|P_n^{(-1,-1)}\|_{\omega^{(-1,-1)}}^2}.
\label{eq:generalized-jacobi-coeff}
\end{equation}

We define the $L^2_{\omega^{(-1,-1)}}$-orthogonal projection
\begin{equation}
\pi_N^{(-1,-1)}:L^2_{\omega^{(-1,-1)}}(-1,1)\to \mathbb{P}_N^0,
\end{equation}
such that
\begin{equation}
\bigl(\pi_N^{(-1,-1)}u-u,v_N\bigr)_{\omega^{(-1,-1)}}=0,
\qquad
\forall\, v_N\in \mathbb{P}_N^0,
\label{eq:projection-generalized-jacobi1}
\end{equation}
or equivalently,
\begin{equation}
(\pi_N^{(-1,-1)}u)(x)=\sum_{n=2}^N \hat{u}_n^{(-1,-1)}P_n^{(-1,-1)}(x),
\label{eq:projection-generalized-jacobi2}
\end{equation}
where $\hat{u}_n^{(-1,-1)}$ is given by \eqref{eq:generalized-jacobi-coeff}.

The projection $\pi_N^{(-1,-1)}$ has the following basic properties.

\begin{proposition}
\label{prop:generalized-jacobi-basic}
For any $u\in L^2_{\omega^{(-1,-1)}}(-1,1)$, the following hold:
\begin{enumerate}
    \item $\pi_N^{(-1,-1)}u$ is the best approximation to $u$ in $L^2_{\omega^{(-1,-1)}}(-1,1)$ over $\mathbb{P}_N^0$, namely,
    \begin{equation}
    \|\pi_N^{(-1,-1)}u-u\|_{\omega^{(-1,-1)}}
    =
    \inf_{v_N\in\mathbb{P}_N^0}\|v_N-u\|_{\omega^{(-1,-1)}}.
    \end{equation}

    \item The projection is contractive in the $L^2_{\omega^{(-1,-1)}}(-1,1)$ norm, and moreover
    \begin{equation}
    \|\pi_N^{(-1,-1)}u\|_{L^2}
    \le
    \|\pi_N^{(-1,-1)}u\|_{\omega^{(-1,-1)}}
    \le
    \|u\|_{\omega^{(-1,-1)}}.
    \end{equation}

    \item If $u\in H_0^1(-1,1)$, then
    \begin{equation}
    \partial_x\pi_N^{(-1,-1)}u
    =
    \pi_{N-1}^{(0,0)}(\partial_x u),
    \label{eq:commuting-property-generalized-jacobi}
    \end{equation}
    and consequently
    \begin{equation}
    \|\partial_x\pi_N^{(-1,-1)}u\|_{L^2}
    \le
    \|\partial_x u\|_{L^2}.
    \end{equation}
    Moreover,
    \begin{equation}
    (\pi_N^{(-1,-1)}u)(x)
    =
    \int_{-1}^x
    \pi_{N-1}^{(0,0)}\bigl(\partial_s u(s)\bigr)\,\mathrm{d}s.
    \end{equation}
\end{enumerate}
\end{proposition}

\begin{proof}
(1) For any $v_N\in\mathbb{P}_N^0$, by \eqref{eq:projection-generalized-jacobi1}, we have
\begin{align}
\|\pi_N^{(-1,-1)}u-u\|_{\omega^{(-1,-1)}}^2
&=
\bigl(\pi_N^{(-1,-1)}u-u,\pi_N^{(-1,-1)}u-u\bigr)_{\omega^{(-1,-1)}} \\
&=
\bigl(\pi_N^{(-1,-1)}u-u,v_N-u\bigr)_{\omega^{(-1,-1)}} \\
&\le
\|\pi_N^{(-1,-1)}u-u\|_{\omega^{(-1,-1)}}\,\|v_N-u\|_{\omega^{(-1,-1)}}.
\end{align}
This proves the first assertion.

(2) Taking $v_N=\pi_N^{(-1,-1)}u$ in \eqref{eq:projection-generalized-jacobi1}, we obtain
\begin{equation}
\|\pi_N^{(-1,-1)}u\|_{\omega^{(-1,-1)}}^2
=
\bigl(\pi_N^{(-1,-1)}u,u\bigr)_{\omega^{(-1,-1)}}
\le
\|\pi_N^{(-1,-1)}u\|_{\omega^{(-1,-1)}}\,\|u\|_{\omega^{(-1,-1)}}.
\end{equation}
Thus
\begin{equation}
\|\pi_N^{(-1,-1)}u\|_{\omega^{(-1,-1)}}\le\|u\|_{\omega^{(-1,-1)}}.
\end{equation}
Moreover, since $(1-x^2)^{-1}\ge 1$ on $(-1,1)$, we also have
\begin{equation}
\|\pi_N^{(-1,-1)}u\|_{L^2}\le\|\pi_N^{(-1,-1)}u\|_{\omega^{(-1,-1)}}.
\end{equation}
This proves the second assertion.

(3) Now assume that $u\in H_0^1(-1,1)$. According to \cite{GSW09},
\begin{equation}
\frac{d}{dx}P_n^{(-1,-1)}(x)
=
-2(n-1)P_{n-1}^{(0,0)}(x).
\label{eq:generalized-jacobi-first-derivative}
\end{equation}
On the other hand, by \eqref{eq:jacobi-derivative},
\begin{equation}
\frac{d^2}{dx^2}P_n^{(-1,-1)}(x)
=
-n(n-1)P_{n-2}^{(1,1)}(x)
=
-n(n-1)\frac{P_n^{(-1,-1)}(x)}{1-x^2}.
\label{eq:generalized-jacobi-second-derivative}
\end{equation}
Hence, since $u(\pm1)=0$, integration by parts gives
\begin{align}
&(u,P_n^{(-1,-1)})_{\omega^{(-1,-1)}}=
\int_{-1}^1 u(x)\frac{P_n^{(-1,-1)}(x)}{1-x^2}\,\mathrm{d}x \\
=&
-\frac{1}{n(n-1)}(u,\partial_x^2 P_n^{(-1,-1)})_{L^2}=
\frac{1}{n(n-1)}(\partial_x u,\partial_x P_n^{(-1,-1)})_{L^2}.
\end{align}
By taking $u=P_n^{(-1,-1)}$ in the above identity, we obtain
\begin{equation}
\|P_n^{(-1,-1)}\|_{\omega^{(-1,-1)}}^2
=
\frac{1}{n(n-1)}
\|\partial_xP_n^{(-1,-1)}\|_{L^2}^2.
\end{equation}
Therefore,
\begin{equation}
\hat{u}_n^{(-1,-1)}
=
\frac{
(\partial_xu,\partial_xP_n^{(-1,-1)})_{L^2}
}{
\|\partial_xP_n^{(-1,-1)}\|_{L^2}^2
}.
\end{equation}
It follows that
\begin{equation}
\partial_x\pi_N^{(-1,-1)}u
=
\sum_{n=2}^N \hat{u}_n^{(-1,-1)}\partial_x P_n^{(-1,-1)}
=
\sum_{n=2}^N
\frac{(\partial_x u,\partial_x P_n^{(-1,-1)})_{L^2}}
{\|\partial_x P_n^{(-1,-1)}\|_{L^2}^2}
\,\partial_x P_n^{(-1,-1)}.
\end{equation}
Using \eqref{eq:generalized-jacobi-first-derivative}, this implies
\begin{equation}
\bigl(
\partial_x\pi_N^{(-1,-1)}u-\partial_xu,
P_m^{(0,0)}
\bigr)_{L^2}
=0,
\qquad m=1,\ldots,N-1.
\end{equation}
For the remaining constant mode, we also have
\begin{equation}
\bigl(
\partial_x\pi_N^{(-1,-1)}u-\partial_xu,
1
\bigr)_{L^2}
=0,
\end{equation}
because both $u$ and $\pi_N^{(-1,-1)}u$ vanish at $x=\pm1$.
Consequently,
\begin{equation}
\bigl(
\partial_x\pi_N^{(-1,-1)}u-\partial_xu,
p
\bigr)_{L^2}
=0,
\qquad
\forall\,p\in\mathbb{P}_{N-1}.
\end{equation}
This proves the commuting property
\begin{equation}
\partial_x\pi_N^{(-1,-1)}u
=
\pi_{N-1}^{(0,0)}(\partial_xu).
\end{equation}
The stability estimate then follows immediately from the $L^2$-stability of the Legendre projection:
\begin{equation}
\|\partial_x\pi_N^{(-1,-1)}u\|_{L^2}
=
\|\pi_{N-1}^{(0,0)}(\partial_xu)\|_{L^2}
\le
\|\partial_xu\|_{L^2}.
\end{equation}
Finally, since $\pi_N^{(-1,-1)}u(-1)=0$, integrating the commuting property from $-1$ to $x$ gives
\begin{equation}
(\pi_N^{(-1,-1)}u)(x)
=
\int_{-1}^x
\partial_s(\pi_N^{(-1,-1)}u)(s)\,\mathrm{d}s
=
\int_{-1}^x
\pi_{N-1}^{(0,0)}\bigl(\partial_su(s)\bigr)\,\mathrm{d}s.
\end{equation}
This completes the proof of the third assertion.
\end{proof}

\begin{remark}
It can be shown that $u\in L^2_{\omega^{(-1,-1)}}(-1,1),\,\partial_x u\in L^2(-1,1)$ if and only if $u\in H_0^1(-1,1)$. Therefore, in item 3, we state the assumption
as $u\in H_0^1(-1,1)$.
\end{remark}

The generalized Jacobi orthogonal projection satisfies an error estimate analogous to that in \prop{jacobi-projection-estimate}.

\begin{proposition}\label{prop:generalized-jacobi-projection-estimate}
For any $u\in B_{-1,-1}^m(-1,1)$, the following hold:
\begin{enumerate}
    \item If $1\le l\le m\le N+1$, then
    \begin{equation}
    \left\|
    \partial_x^l\bigl(\pi_N^{(-1,-1)}u-u\bigr)
    \right\|_{\omega^{(l-1,l-1)}}
    \lesssim
    N^{\,l-m}
    \left\|
    \partial_x^m u
    \right\|_{\omega^{(m-1,m-1)}}.
    \label{eq:generalized-jacobi-projection-estimate-1}
    \end{equation}

    \item If $m>N+1$ and $1\le l\le N+1$, then
    \begin{equation}
    \left\|
    \partial_x^l\bigl(\pi_N^{(-1,-1)}u-u\bigr)
    \right\|_{\omega^{(l-1,l-1)}}
    \lesssim
    \left(\frac{\sqrt{e/2}}{N}\right)^{N-l+\frac54}
    \left\|
    \partial_x^{N+1}u
    \right\|_{\omega^{(N,N)}}.
    \label{eq:generalized-jacobi-projection-estimate-2}
    \end{equation}
\end{enumerate}
\end{proposition}

\begin{proof}
By item 3 of \prop{generalized-jacobi-basic}, we have
\begin{equation}
\left\|
\partial_x^l\bigl(\pi_N^{(-1,-1)}u-u\bigr)
\right\|_{\omega^{(l-1,l-1)}}
=
\left\|
\partial_x^{l-1}\bigl(\pi_{N-1}^{(0,0)}\partial_xu-\partial_xu\bigr)
\right\|_{\omega^{(l-1,l-1)}}.
\label{eq:generalized-jacobi-reduction}
\end{equation}
Noting that $\partial_xu\in B_{0,0}^{m-1}(-1,1)$, the conclusion then follows immediately from \prop{jacobi-projection-estimate}, applied to $\partial_xu$ with $\alpha=\beta=0$, $N$ replaced by $N-1$, $l$ replaced by $l-1$, and $m$ replaced by $m-1$.
\end{proof}

\begin{remark}
According to \eqref{eq:generalized-jacobi-reduction}, and analogously to the remark after \prop{jacobi-projection-estimate}, if $u\in B_{-1,-1}^N(-1,1)$, then for $1\le l\le N$ we have
\begin{equation}
\left\|
\partial_x^l\bigl(\pi_N^{(-1,-1)}u-u\bigr)
\right\|_{\omega^{(l-1,l-1)}}
\lesssim
\left(\frac{\sqrt{e/2}}{N}\right)^{N-l+\frac34}
\left\|
\partial_x^N u
\right\|_{\omega^{(N-1,N-1)}}.
\label{eq:generalized-jacobi-projection-estimate-N}
\end{equation}
\end{remark}

Finally, we derive an improved $L^2$-error estimate for the generalized Jacobi orthogonal projection.

\begin{proposition}\label{prop:generalized-jacobi-L2-estimate}
Let $u\in H_0^1(-1,1)$. Then
\begin{equation}
\|\pi_N^{(-1,-1)}u-u\|_{L^2}
\lesssim
N^{-1}\|\partial_x(\pi_N^{(-1,-1)}u-u)\|_{L^2}.
\end{equation}
\end{proposition}

\begin{proof}
Let
\begin{equation}
v_N:=\pi_N^{(-1,-1)}u
=
\int_{-1}^x \pi_{N-1}^{(0,0)}\bigl(\partial_su(s)\bigr)\,\mathrm{d}s.
\end{equation}
Consider the boundary value problem
\begin{equation}
-\partial_x^2 z=v_N-u,
\qquad
z(\pm1)=0.
\end{equation}
Then
\begin{align}
&\|v_N-u\|_{L^2}^2=(v_N-u,-\partial_x^2z)_{L^2}=(\partial_x(v_N-u),\partial_x z)_{L^2}\\
=&(\pi_{N-1}^{(0,0)}\partial_xu-\partial_xu,\partial_x z)_{L^2}=(\pi_{N-1}^{(0,0)}\partial_xu-\partial_xu,\partial_x z-\pi_{N-1}^{(0,0)}\partial_x z)_{L^2}\\
=&(\partial_x v_N-\partial_xu,\partial_x z-\partial_x\pi_N^{(-1,-1)}z)_{L^2}
\end{align}
where the penultimate equality follows from the orthogonality property of the projection $\pi_{N-1}^{(0,0)}$ in $L^2$, and in the last equality we used item 3 of \prop{generalized-jacobi-basic}, namely, $\partial_x\pi_N^{(-1,-1)}z=\pi_{N-1}^{(0,0)}\partial_x z$. Applying \prop{generalized-jacobi-projection-estimate} with $l=1,m=2$, we obtain
\begin{equation}
\|\partial_xz-\partial_x\pi_N^{(-1,-1)}z\|_{L^2} 
\lesssim
N^{-1}\|\partial_x^2 z\|_{\omega^{(1,1)}}\le
N^{-1}\|\partial_x^2 z\|_{L^2}=
N^{-1}\|v_N-u\|_{L^2}.
\end{equation}
It follows that
\begin{align}
\|v_N-u\|_{L^2}^2
&\le
\|\partial_x(v_N-u)\|_{L^2}\,
\|\partial_x z-\pi_{N-1}^{(0,0)}\partial_x z\|_{L^2} \\
&\lesssim
N^{-1}\|\partial_x(v_N-u)\|_{L^2}\,\|v_N-u\|_{L^2}.
\end{align}
Dividing both sides by $\|v_N-u\|_{L^2}$, we conclude that
\begin{equation}
\|v_N-u\|_{L^2}
\lesssim
N^{-1}\|\partial_x(v_N-u)\|_{L^2}.
\end{equation}
Since $v_N=\pi_N^{(-1,-1)}u$, the proof is complete.
\end{proof}

\begin{remark}
If one applies the Cauchy--Schwarz inequality directly, then
\begin{align}
\|\pi_N^{(-1,-1)}u-u\|_{L^2}^2
&=
\int_{-1}^1
\left(
\int_{-1}^x
\bigl(\pi_{N-1}^{(0,0)}\partial_su-\partial_su\bigr)\,\mathrm{d}s
\right)^2
\mathrm{d}x \\
&\le
\int_{-1}^1
(x+1)
\int_{-1}^x
\bigl(\pi_{N-1}^{(0,0)}\partial_su-\partial_su\bigr)^2\,\mathrm{d}s\,\mathrm{d}x \\
&\le
\int_{-1}^1
(x+1)
\int_{-1}^1
\bigl(\pi_{N-1}^{(0,0)}\partial_su-\partial_su\bigr)^2\,\mathrm{d}s\,\mathrm{d}x \\
&=
2\|\pi_{N-1}^{(0,0)}\partial_xu-\partial_xu\|_{L^2}^2.
\end{align}
By item 3 of \prop{generalized-jacobi-basic}, we have
\begin{equation}
    \|\pi_N^{(-1,-1)}u-u\|_{L^2}\lesssim\|\partial_x(\pi_N^{(-1,-1)}u-u)\|_{L^2}.
\end{equation}
This estimate is weaker than that given in \prop{generalized-jacobi-L2-estimate}.
\end{remark}

\subsection{Error estimates for the tensor-product generalized Jacobi projection}

We now extend the generalized Jacobi orthogonal projection to the $d$-dimensional tensor-product setting and derive corresponding approximation estimates, which will be used in the error analysis of the spectral-Galerkin solution.

\begin{proposition}\label{prop:tensor-generalized-jacobi-H1}
Let
\begin{equation}
\Pi_N^{(-1,-1)}:=\pi_{N,1}^{(-1,-1)}\pi_{N,2}^{(-1,-1)}\cdots \pi_{N,d}^{(-1,-1)},
\end{equation}
where $\pi_{N,j}^{(-1,-1)}$ denotes the generalized Jacobi orthogonal projection acting on the $j$th variable. For each $i,j\in\{1,\dots,d\}$, define
\begin{equation}
W_{i,j}(\boldsymbol{x})
:=
\prod_{\substack{1\le k\le j-1\\ k\neq i}}
\omega^{(-1,-1)}(x_k),
\end{equation}
and for each integer $s\ge 2$, define
\begin{equation}
W_{i,j}^{(s)}(\boldsymbol{x})
:=
\begin{cases}
W_{i,i}(\boldsymbol{x})\,\omega^{(s-1,s-1)}(x_i), & j=i,\\[4pt]
W_{i,j}(\boldsymbol{x})\,\omega^{(s-2,s-2)}(x_j), & j\neq i.
\end{cases}
\label{eq:tensor-generalized-jacobi-weight}
\end{equation}
Assume that all weighted norms appearing on the right-hand sides below are finite. Then the following hold:
\begin{enumerate}
    \item If $2\le m\le N+1$, then
    \begin{equation}
    \|\nabla(\Pi_N^{(-1,-1)}u-u)\|_{L^2}
    \lesssim
    N^{1-m}
    \left(
    \sum_{i=1}^d
    \left(
    \sum_{j=1}^d
    \|\partial_{x_j}^{\,m-1}\partial_{x_i}u\|_{W_{i,j}^{(m)}}
    \right)^2
    \right)^{\frac12}.
    \label{eq:tensor-generalized-jacobi-H1}
    \end{equation}

    \item If $m>N+1$, then
    \begin{equation}
    \|\nabla(\Pi_N^{(-1,-1)}u-u)\|_{L^2}
    \lesssim
    \left(\frac{\sqrt{e/2}}{N}\right)^{N+\frac14}
    \left(
    \sum_{i=1}^d
    \left(
    \sum_{j=1}^d
    \|\partial_{x_j}^{\,N}\partial_{x_i}u\|_{W_{i,j}^{(N+1)}}
    \right)^2
    \right)^{\frac12}.
    \label{eq:tensor-generalized-jacobi-H1-analytic}
    \end{equation}
\end{enumerate}
\end{proposition}

\begin{proof}
For simplicity, write $I_j:=\pi_{N,j}^{(-1,-1)}$, $j=1,\dots,d$. Fix $i\in\{1,\dots,d\}$. Since $\partial_{x_i}$ commutes with $I_j$ for $j< i$, we have
\begin{align}
\partial_{x_i}(\Pi_N^{(-1,-1)}u-u)
&=
\sum_{j=1}^{i-1}
I_1\cdots I_{j-1}(I_j-I)\partial_{x_i}u
+
I_1\cdots I_{i-1}\partial_{x_i}(I_i-I)u \nonumber\\
&\quad+
\sum_{j=i+1}^{d}
I_1\cdots I_{i-1}\partial_{x_i}I_i\,I_{i+1}\cdots I_{j-1}(I_j-I)u.
\label{eq:tensor-generalized-jacobi-decomposition}
\end{align}
Using item 2 and item 3 of \prop{generalized-jacobi-basic}, together with Fubini's theorem, we obtain
\begin{equation}
\|\partial_{x_i}(\Pi_N^{(-1,-1)}u-u)\|_{L^2}
\le
\|\partial_{x_i}(I_i-I)u\|_{W_{i,i}}
+
\sum_{j\neq i}\|(I_j-I)\partial_{x_i}u\|_{W_{i,j}}.
\label{eq:tensor-generalized-jacobi-reduced}
\end{equation}
For $j\neq i$, \prop{generalized-jacobi-L2-estimate} gives
\begin{equation}
\|(I_j-I)\partial_{x_i}u\|_{W_{i,j}}
\lesssim
N^{-1}
\|\partial_{x_j}(I_j-I)\partial_{x_i}u\|_{W_{i,j}}.
\label{eq:tensor-generalized-jacobi-L2-j}
\end{equation}
Hence
\begin{equation}
\|\partial_{x_i}(\Pi_N^{(-1,-1)}u-u)\|_{L^2}
\lesssim
\|\partial_{x_i}(I_i-I)u\|_{W_{i,i}}
+
N^{-1}\sum_{j\neq i}
\|\partial_{x_j}(I_j-I)\partial_{x_i}u\|_{W_{i,j}}.
\label{eq:tensor-generalized-jacobi-before-final}
\end{equation}

If $2\le m\le N+1$, then, by \prop{generalized-jacobi-projection-estimate}, we obtain
\begin{align}
\|\partial_{x_i}(\Pi_N^{(-1,-1)}u-u)\|_{L^2}
&\lesssim
N^{1-m}\|\partial_{x_i}^{m}u\|_{W_{i,i}^{(m)}}
+
N^{-1}\sum_{j\neq i}
N^{1-(m-1)}
\|\partial_{x_j}^{\,m-1}\partial_{x_i}u\|_{W_{i,j}^{(m)}} \nonumber\\
&=
N^{1-m}
\sum_{j=1}^d
\|\partial_{x_j}^{\,m-1}\partial_{x_i}u\|_{W_{i,j}^{(m)}}.
\end{align}
Therefore,
\begin{align}
\|\nabla(\Pi_N^{(-1,-1)}u-u)\|_{L^2}
=
\left(
\sum_{i=1}^d
\|\partial_{x_i}(\Pi_N^{(-1,-1)}u-u)\|_{L^2}^2
\right)^{\frac12}\\
\lesssim N^{1-m}\left(\sum_{i=1}^d\left(\sum_{j=1}^d\|\partial_{x_j}^{\,m-1}\partial_{x_i}u\|_{W_{i,j}^{(m)}}\right)^2\right)^{\frac12}.
\end{align}

If $m>N+1$, then one proceeds similarly. More precisely, the first term is estimated by the second estimate in \prop{generalized-jacobi-projection-estimate}, whereas the remaining terms are estimated by the remark after \prop{generalized-jacobi-projection-estimate}. Thus
\begin{align}
\|\partial_{x_i}(\Pi_N^{(-1,-1)}u-u)\|_{L^2}
&\lesssim
\left(\frac{\sqrt{e/2}}{N}\right)^{N+\frac14}
\|\partial_{x_i}^{N+1}u\|_{W_{i,i}^{(N+1)}}
+
N^{-1}\sum_{j\neq i}
\left(\frac{\sqrt{e/2}}{N}\right)^{N-\frac14}
\|\partial_{x_j}^{N}\partial_{x_i}u\|_{W_{i,j}^{(N+1)}} \nonumber\\
&\lesssim
\left(\frac{\sqrt{e/2}}{N}\right)^{N+\frac14}
\sum_{j=1}^d
\|\partial_{x_j}^{N}\partial_{x_i}u\|_{W_{i,j}^{(N+1)}}.
\end{align}
Therefore, 
\begin{equation}
    \|\nabla(\Pi_N^{(-1,-1)}u-u)\|_{L^2}
    \lesssim
    \left(\frac{\sqrt{e/2}}{N}\right)^{N+\frac14}
    \left(
    \sum_{i=1}^d
    \left(
    \sum_{j=1}^d
    \|\partial_{x_j}^{\,N}\partial_{x_i}u\|_{W_{i,j}^{(N+1)}}
    \right)^2
    \right)^{\frac12}.
\end{equation}
This completes the proof.
\end{proof}

As a consequence of \eqref{eq:H1-error-reduction-recall} and \prop{tensor-generalized-jacobi-H1}, we obtain the following error estimate for the spectral-Galerkin solution.

\begin{corollary}\label{cor:galerkin-error-analytic}
Consider an instance of \prb{pde}, namely, the $d$-dimensional second-order linear elliptic boundary value problem with homogeneous Dirichlet boundary conditions \eqref{eq:model-problem}. Let the approximation space be generated by the tensor-product basis associated with the one-dimensional basis \eqref{eq:new-basis-1d}. Assume that the exact solution $u$ is analytic, and define
\begin{equation}
C_N
:=
\left(
\sum_{i=1}^d
\left(
\sum_{j=1}^d
\|\partial_{x_j}^{\,N+1}\partial_{x_i}u\|_{W_{i,j}^{(N+2)}}
\right)^2
\right)^{\frac12}
\end{equation}
where $W_{i,j}^{(s)}$ is defined in \eqref{eq:tensor-generalized-jacobi-weight}. Then the spectral-Galerkin solution $u_N$ satisfies
\begin{equation}
\|u_N-u\|_{H^1}
\lesssim
C_N\kappa(A)^{\frac12}
\left(\frac{\sqrt{e/2}}{N+1}\right)^{N+\frac54}.
\end{equation}
Moreover,
\begin{equation}
\|u_N-u\|_{L^2}
\lesssim
C_N\kappa(A)^{\frac12}d^{-\frac12}
\left(\frac{\sqrt{e/2}}{N+1}\right)^{N+\frac54}.
\end{equation}
\end{corollary}

\begin{proof}
Since the approximation space satisfies $V_N=(\mathbb{P}_{N+1}^0)^{\otimes d}$, we may choose
\begin{equation}
v_N:=\Pi_{N+1}^{(-1,-1)}u\in V_N
\end{equation}
in \eqref{eq:H1-error-reduction-recall}. \prop{tensor-generalized-jacobi-H1} then yields
\begin{equation}
\|\nabla(v_N-u)\|_{L^2}
\lesssim
C_N\left(\frac{\sqrt{e/2}}{N+1}\right)^{N+\frac54},
\end{equation}
and hence
\begin{equation}
\|u_N-u\|_{H^1}
\lesssim
C_N\kappa(A)^{\frac12}
\left(\frac{\sqrt{e/2}}{N+1}\right)^{N+\frac54}.
\end{equation}
The $L^2$-estimate follows from the Poincar\'e inequality applied to $u_N-u\in H_0^1(\Omega)$.
\end{proof}
\section{An illustrative example of the quantum spectral-Galerkin method}
\label{app:example}

In \sec{jacobi-spectral-Galerkin method}, we introduced the spectral-Galerkin discretization for the model problem in \sec{model problem and algorithmic objectives}, and in \sec{quantum spectral-Galerkin method} we developed the corresponding quantum algorithm. In this appendix, we present a concrete instance of this construction and report numerical results that illustrate the behavior of the method and corroborate our theoretical analysis.

We consider the boundary value problem
\begin{equation}
\L u
=
-\sum_{p,q=1}^{d} A_{pq}\,\partial_{x_p}\partial_{x_q}u
=
f
\quad \text{in } \Omega:=(-1,1)^d,
\quad
u|_{\partial\Omega}=0,
\label{eq:example-pde}
\end{equation}
where $A=(A_{pq})_{p,q=1}^d\in\mathbb{R}^{d\times d}$ is a constant symmetric positive-definite matrix.

The Galerkin approximation is written as
\begin{equation}
u_N(\boldsymbol{x})
=
\sum_{\boldsymbol{k}\in [N]_0^d}
c_{\boldsymbol{k}}\,\Phi_{\boldsymbol{k}}(\boldsymbol{x}),
\label{eq:example-uN}
\end{equation}
where the basis functions are
\begin{equation}
\Phi_{\boldsymbol{k}}(\boldsymbol{x})
=
\prod_{j=1}^d \phi_{k_j}(x_j),
\qquad
\boldsymbol{k}=(k_1,\dots,k_d)\in [N]_0^d,
\label{eq:example-tensor-basis}
\end{equation}
and
\begin{equation}
\phi_k(x)
=
\frac{1}{\sqrt{\gamma_k^{(2,2)}}}(1-x^2)P_k^{(2,2)}(x),
\qquad
k\in [N]_0,
\label{eq:example-phi}
\end{equation}
where $\gamma_k^{(2,2)}$ is the normalization constant in \eqref{eq:jacobi-norm}. The corresponding load vector is
\begin{equation}
\boldsymbol{f}
=
\sum_{\boldsymbol{l}\in [N]_0^d}
\langle f,\Phi_{\boldsymbol{l}}\rangle\ket{\boldsymbol{l}}.
\end{equation}
We assume access to an oracle that prepares its normalized quantum state $\ket{f}=\boldsymbol{f}/\|\boldsymbol{f}\|$. We then apply the quantum linear systems algorithm to solve the linear system
\begin{equation}
L\boldsymbol{c}=\boldsymbol{f},
\label{eq:example-linear-system}
\end{equation}
where the solution vector is
\begin{equation}
\boldsymbol{c}
=
\sum_{\boldsymbol{k}\in [N]_0^d}
c_{\boldsymbol{k}}\ket{\boldsymbol{k}}.
\label{eq:example-c}
\end{equation}
Moreover, $L=L_1+L_2$, with
\begin{align}
L_1
&=
\sum_{p=1}^d
A_{pp}
\bigl(
I^{\otimes(p-1)}
\otimes
S
\otimes
I^{\otimes(d-p)}
\bigr),
\label{eq:example-L1}
\\
L_2
&=
-2\sum_{1\leq p<q\leq d}
A_{pq}
\bigl(
I^{\otimes(p-1)}
\otimes
T
\otimes
I^{\otimes(q-p-1)}
\otimes
T
\otimes
I^{\otimes(d-q)}
\bigr),
\label{eq:example-L2}
\end{align}
where
\begin{equation}
S=\sum_{j,k\in[N]_0}\langle\phi_k',\phi_j'\rangle\ket{j}\bra{k},
\qquad
T=\sum_{j,k\in[N]_0}\langle\phi_k',\phi_j\rangle\ket{j}\bra{k}.
\label{eq:example-ST}
\end{equation}
Solving \eqref{eq:example-linear-system} yields the coefficient vector $\boldsymbol{c}$, and hence the Galerkin approximation $u_N$.

We next illustrate the condition number bound in \lem{S-cond-N4} for the one-dimensional stiffness matrix $S$. Recall that \lem{S-cond-N4} implies
\begin{equation}
\kappa(S)\le \frac{1}{\pi^2}(N+2)^2(N+3)^2.
\label{eq:appendix-cond-S-bound}
\end{equation}
\fig{cond-S} compares the computed condition number $\kappa(S)$ for $50 \leq N \leq 300$ with the upper bound. In all tested cases, the computed values remain below the theoretical bound. Moreover, a least-squares fit of the form $CN^p$ yields an exponent $p$ close to $4$, indicating that the growth of $\kappa(S)$ is consistent with the estimate $\kappa(S)=O(N^4)$.

\begin{figure}[t]
    \centering
    \includegraphics[width=0.8\linewidth]{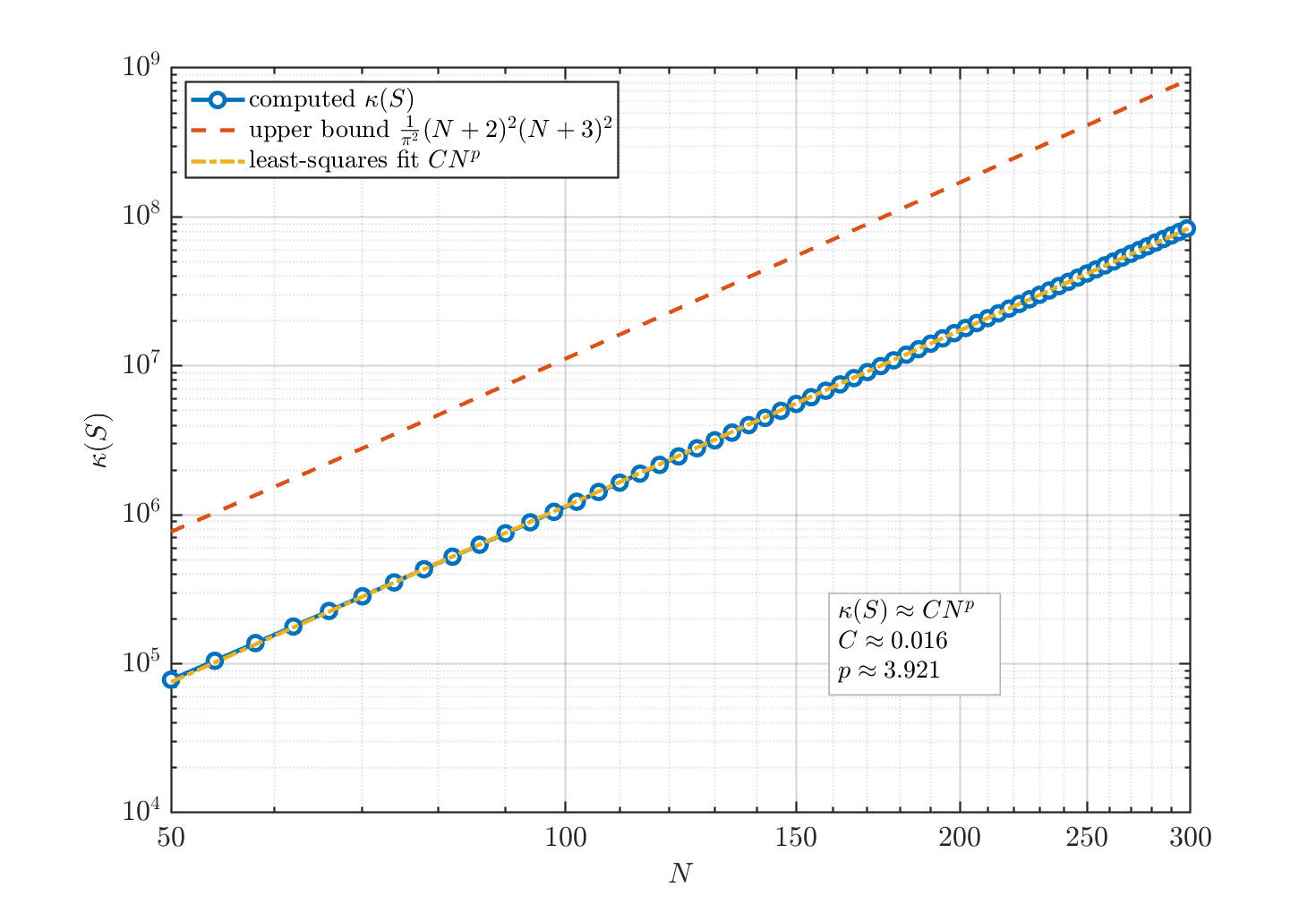}
    \caption{Condition number of the one-dimensional stiffness matrix $S$ as a function of $N$ for $50\leq N\leq 300$. The dashed curve is the theoretical upper bound $\frac{1}{\pi^2}(N+2)^2(N+3)^2$, and the dash-dotted curve is a least-squares fit of the form $CN^p$.}
    \label{fig:cond-S}
\end{figure}

Finally, we present numerical experiments for the Galerkin system associated with the basis \eqref{eq:example-phi}. We consider the model problem \eqref{eq:example-pde} in dimensions $d=2,3$, and test several symmetric positive-definite matrices $A$ with different condition numbers. As the exact solution, we take
\begin{equation}
u(\boldsymbol{x})=
\exp\Bigl(\sum_{j=1}^d x_j\Bigr)
\Bigl(1+\sum_{1\le p<q\le d}x_p x_q\Bigr)
\prod_{j=1}^d (1-x_j^2).
\label{eq:exact-solution-numerics}
\end{equation}
This function satisfies the homogeneous Dirichlet boundary condition. Moreover, for $d\ge 2$, it is not separable in the spatial variables. The right-hand side $f$ is then determined from \eqref{eq:example-pde} by direct substitution.

The left panels of Figures~\ref{fig:numerics-d2}--\ref{fig:numerics-d3} show the relative $L^2$ error of the numerical solution as a function of $N$. The right panels show the ratio $\kappa(L)/(\kappa(A)\kappa(S))$, which is used to examine the bound
\begin{equation}
\kappa(L)\le \kappa(A)\kappa(S).
\end{equation}
proved in \lem{jacobi-basis-properties}. The results show rapid error decay as $N$ increases, and the ratio $\kappa(L)/(\kappa(A)\kappa(S))$ stays below $1$ in all tested cases, in agreement with the theoretical bound.


\begin{figure}[t]
    \centering
    \includegraphics[width=0.9\linewidth]{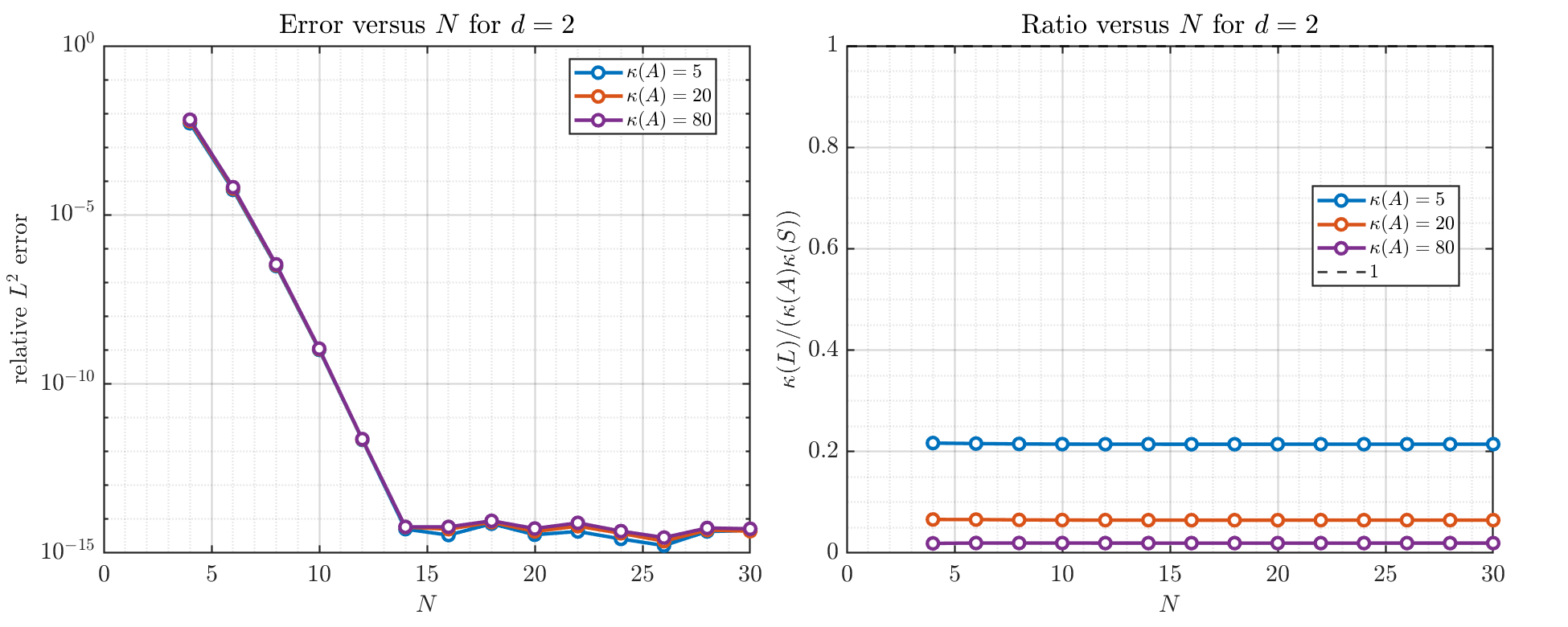}
    \caption{Numerical results for the case $d=2$. Left: relative $L^2$ error versus $N$. Right: the ratio $\kappa(L)/(\kappa(A)\kappa(S))$ versus $N$.}
    \label{fig:numerics-d2}
\end{figure}

\begin{figure}[t]
    \centering
    \includegraphics[width=0.9\linewidth]{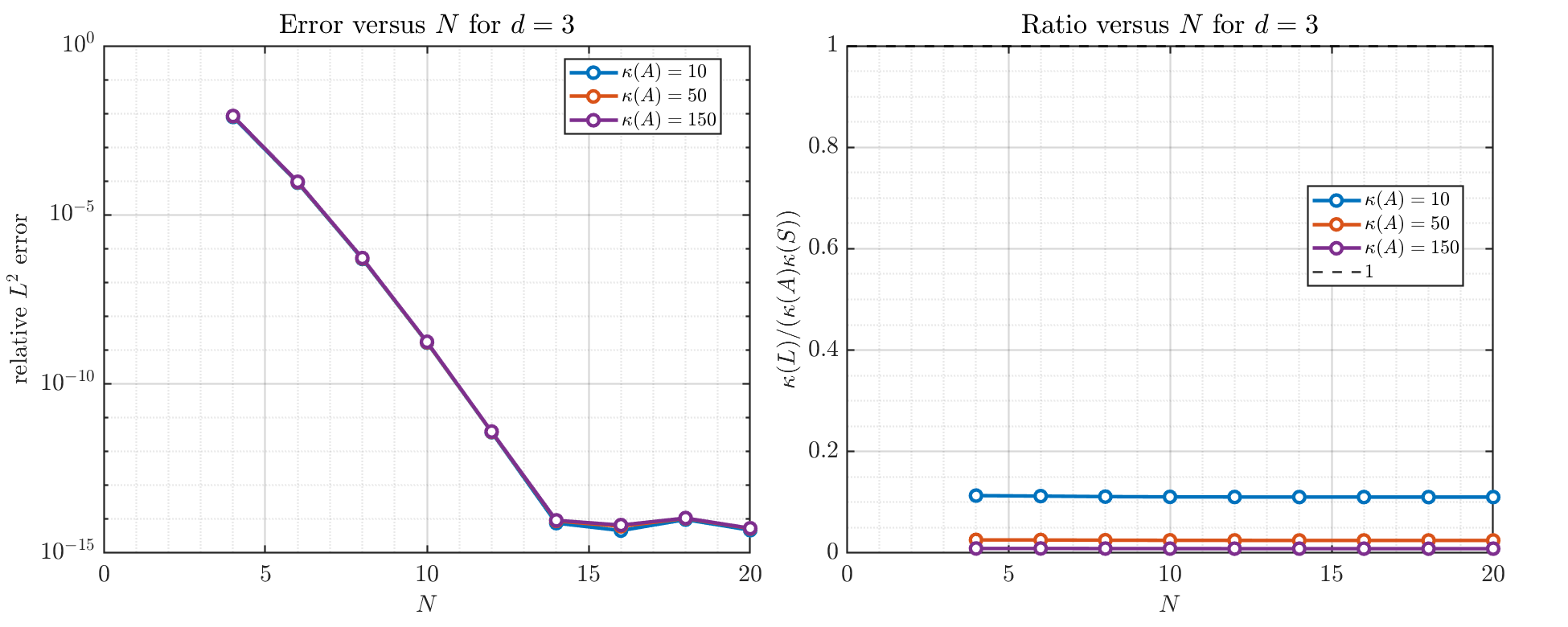}
    \caption{Numerical results for the case $d=3$. Left: relative $L^2$ error versus $N$. Right: the ratio $\kappa(L)/(\kappa(A)\kappa(S))$ versus $N$.}
    \label{fig:numerics-d3}
\end{figure}

\end{document}